\documentclass[10pt, onecolumn]{IEEEtran}

\usepackage{amsmath}
\usepackage{amssymb}
\usepackage{amsfonts}
\usepackage{epsfig}
\usepackage{subfigure}
\usepackage{calc}
\usepackage{color}
\usepackage[all]{xy}
\usepackage{bm}
\usepackage{algorithmicx}
\usepackage{algpseudocode}
\usepackage{algorithm}
\usepackage{enumerate}
\usepackage{pdfsync}
\usepackage{framed}

\newtheorem{defn}{Definition}

\newtheorem{thm}{Theorem}[section]
\newtheorem{cor}[thm]{Corollary}
\newtheorem{prop}{Proposition}

\newtheorem{lem}[thm]{Lemma}
\newtheorem{conj}[thm]{Conjecture}
\newtheorem{constr}[thm]{Construction}
\newtheorem{note}{Remark}

\newcommand{\bit}{\begin{itemize}}
\newcommand{\eit}{\end{itemize}}
\newcommand{\bcor}{\begin{cor}}
\newcommand{\ecor}{\end{cor}}
\newcommand{\beq}{\begin{equation}}
\newcommand{\eeq}{\end{equation}}
\newcommand{\beqn}{\begin{equation*}}
\newcommand{\eeqn}{\end{equation*}}
\newcommand{\bea}{\begin{eqnarray}}
\newcommand{\eea}{\end{eqnarray}}
\newcommand{\bean}{\begin{eqnarray*}}
\newcommand{\eean}{\end{eqnarray*}}
\newcommand{\ben}{\begin{enumerate}}
\newcommand{\een}{\end{enumerate}}
\newcommand{\bdefn}{\begin{defn}}
\newcommand{\edefn}{\end{defn}}
\newcommand{\bnote}{\begin{note}}
\newcommand{\enote}{\end{note}}
\newcommand{\bprop}{\begin{prop}}
\newcommand{\eprop}{\end{prop}}
\newcommand{\blem}{\begin{lem}}
\newcommand{\elem}{\end{lem}}
\newcommand{\bthm}{\begin{thm}}
\newcommand{\ethm}{\end{thm}}
\newcommand{\bconj}{\begin{conj}}
\newcommand{\econj}{\end{conj}}
\newcommand{\bconstr}{\begin{constr}}
\newcommand{\econstr}{\end{constr}}
\newcommand{\bpf}{\begin{proof}}
\newcommand{\epf}{\end{proof}}

\title{The Storage-Repair-Bandwidth Trade-off of Exact Repair Linear Regenerating Codes for the Case $d = k = n-1$}
\author{N. Prakash and M. Nikhil Krishnan  
	\thanks{N. Prakash and M. Nikhil Krishnan are with the Department of ECE, Indian Institute of Science, Bangalore, 560012, India (email: \{prakashn, nikhilkm\}@ece.iisc.ernet.in).}
\thanks{N. Prakash was an intern at NetApp, Bangalore for a part of the duration of this work. } }

\begin{document}


\maketitle

\begin{abstract}
In this paper, we consider the setting of exact repair linear regenerating codes. Under this setting, we derive a new  outer bound on the storage-repair-bandwidth trade-off for the case when $d = k = n -1$, where $(n, k, d)$ are parameters of the regenerating code, with their usual meaning. Taken together with the achievability result of Tian et. al. \cite{BirenPVK_canonical}, we show that the new outer bound derived here completely characterizes the trade-off for the case of exact repair linear regenerating codes, when $d = k = n -1$. The new outer bound is derived by analyzing the dual code of the linear regenerating code.
\end{abstract}

\section{Introduction}\label{sec:intro}

In the regenerating-code framework, a file of size $B$ symbols is encoded into $n\alpha$ symbols and distributed among $n$ nodes in the network, such that each node stores $\alpha$ symbols. These symbols are assumed to be drawn from a finite field $\mathbb{F}_q$. The property of data collection demands that one should be able to recover the entire uncoded file by connecting to any $k$ nodes (see Fig. \ref{fig:intro_regen_framework}) and downloading all the $k\alpha$ coded symbols in them. Further, repair of a single failed node is required to be accomplished by connecting to any $d$ surviving nodes and downloading $\beta \leq \alpha$ symbols from each node. The quantity $d\beta$ is termed as the repair-bandwidth. Two notions of node repair exist, and these are known as functional repair and exact repair.  Under functional repair, the code symbols in the replacement node are such that data collection and node repair properties continue to hold. Under exact repair, the contents of the failed and replacement nodes are identical. 
\begin{figure}[h]
  \centering
\includegraphics[width=5in]{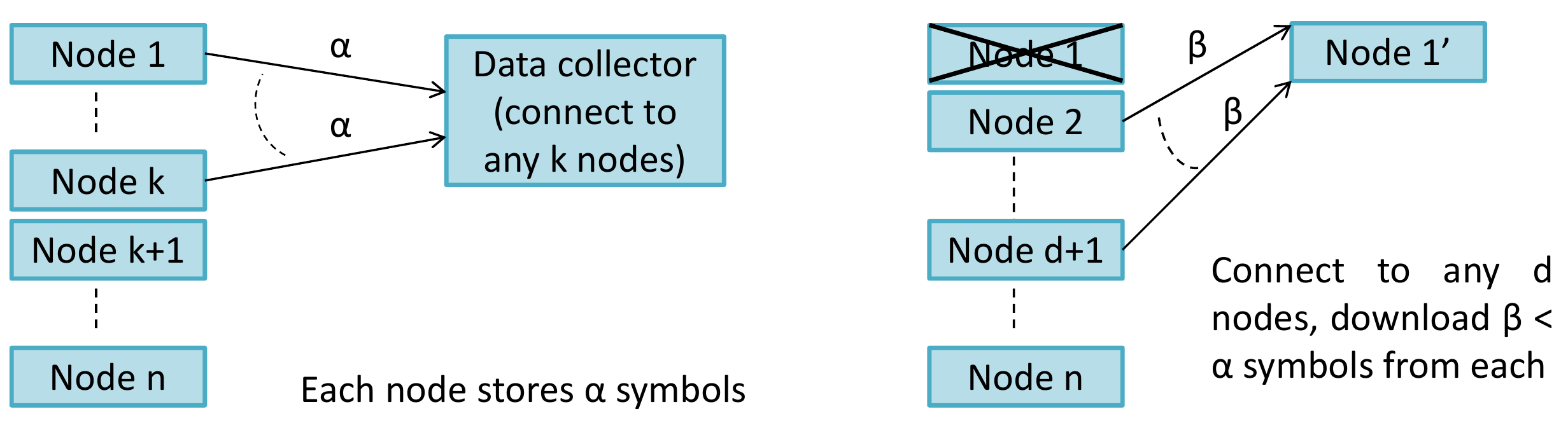}
  \caption{The Regenerating Code Framework.}
  \label{fig:intro_regen_framework}
\end{figure}
A cut-set bound argument based on network-coding was used in \cite{DimGodWuWaiRam} to show that under the framework of functional repair (FR), the file size $B$ is upper bounded by 
\bea \label{eq:intro_cut_set_bd}
B & \leq & \sum_{i=0}^{k-1} \min\{\alpha,(d-i)\beta\}.
\eea
For fixed values of parameters $\{B,k,d\}$, there are multiple pairs $(\alpha,\beta)$ that satisfy \eqref{eq:intro_cut_set_bd} with equality. This leads to the  storage-repair-bandwidth trade-off shown in Fig.~\ref{fig:intro_trade-off} which is piece-wise linear. Existence of FR regenerating codes that can achieve any point on the storage-repair-bandwidth trade-off was also shown in \cite{DimGodWuWaiRam}. The two extremal points on the trade-off curve are termed as the Minimum Storage Regeneration (MSR) and Minimum Bandwidth Regeneration (MBR) points. At the MSR point, the total storage-overhead is as small as possible, while at the MBR point, the repair-bandwidth is the least.  The intermediate points on the curve will be referred to as FR-interior points. 
\begin{figure}[h]
  \centering
\includegraphics[width=3in]{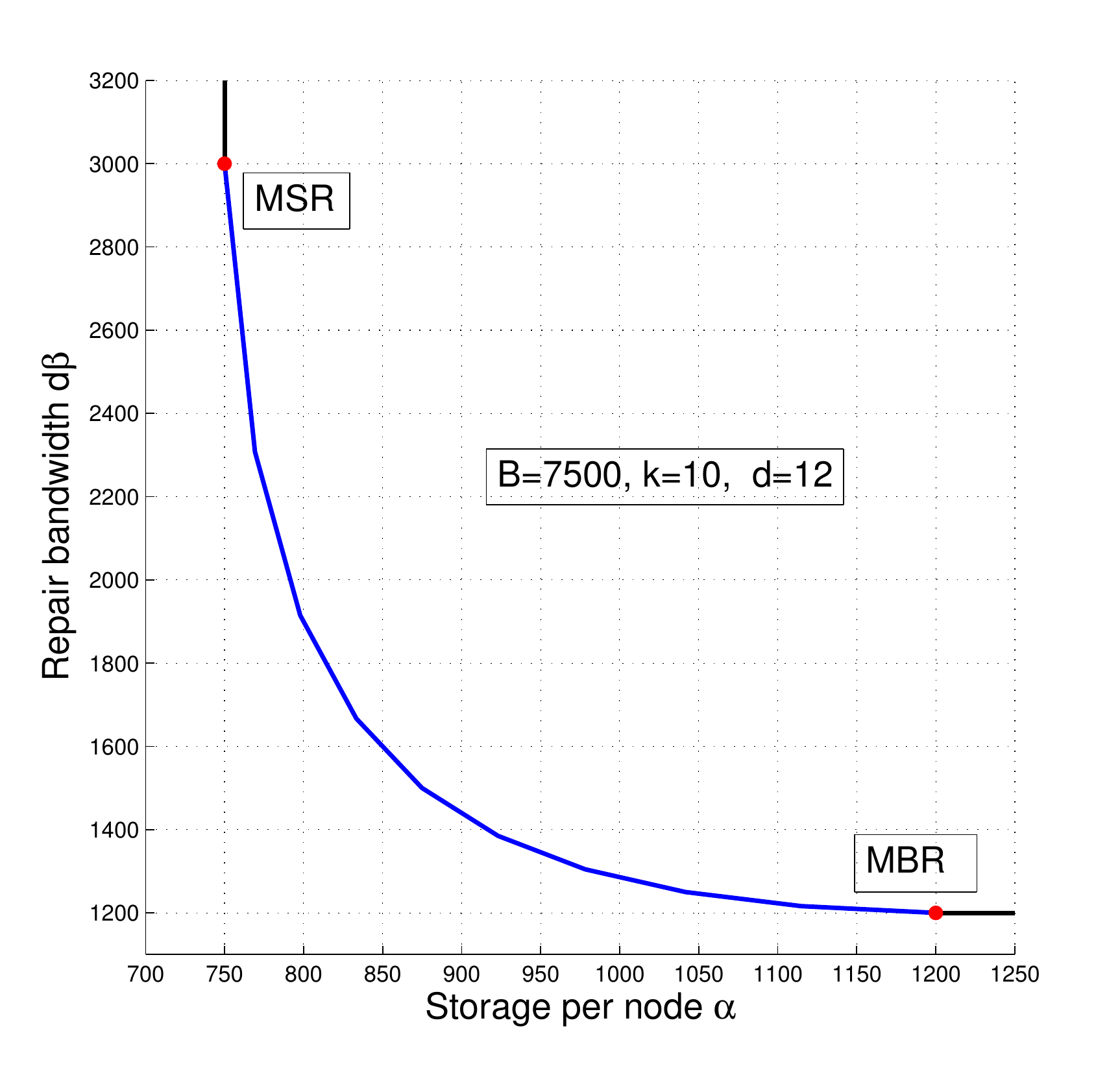}
  \caption{The Storage-Overhead Repair-Bandwidth Trade-off for Regenerating Codes for an example set of parameters.}
  \label{fig:intro_trade-off}
\end{figure}
Several constructions of MSR and MBR codes, having the property of exact repair (ER) exist in literature. Explicit constructions of ER MSR codes for a class of parameters are presented in  \cite{SuhRam,RasShaKum_pm,PapDimCad,ShaRasKumRam_ia,TamWanBru}, whereas the existence of ER MSR codes for all $(n,k,d), \ n > d \geq k$ is shown in \cite{CadJafMalRamSuh}. Explicit ER MBR codes for  all $(n,k,d), \ n > d \geq k$ are presented in \cite{RasShaKum_pm}.  In \cite{ShaRasKumRam_rbt}, a class of ER MBR codes with $d=(n-1)$ is presented, and these codes are termed as repair-by-transfer MBR codes as they enable node repair without need for any operation other than simple data transfer. 

\subsection{The Trade-off for the Case of Exact Repair}

Following results are known in literature regarding the storage-repair-bandwidth trade-off for the case of ER regenerating codes.

\begin{enumerate}[1.]
\item The non-existence of ER regenerating codes which operate on the FR-interior points of the trade-off curve (with the possible exception of the line segment from the MSR point to the next deflection point) was shown in \cite{ShaRasKumRam_rbt}.
\item The trade-off of ER regenerating codes with parameters $(n=4, k=3,d=3)$ was characterized in \cite{Tia}. Except for a region near the MSR point, the interior points on ER trade-off (to be abbreviated as ER-interior points) for the case $(n=4, k=3,d=3)$ lie strictly away from the FR-interior points. 
\item In \cite{SasSenKum}, an  outer bound on the ER trade-off for any general $(n, k, d), n \geq 4$ was derived, which  established that the ER-interior points for any $(n, k, d), n \geq 5$ also lie strictly away from the corresponding FR-interior points (except possibly for a small region near the MSR point). For the case of $(4, 3, 3)$, the bound in \cite{SasSenKum} coincided with the bound in \cite{Tia}. Further, it is also known from the results in \cite{BirenITW} that the bound in \cite{SasSenKum} is optimal when the parameters of the ER regenerating code are given by $(n, k=3, d=n-1)$. However, when $k \geq 4$, the optimality of the outer bound in \cite{SasSenKum} is not known in general. 
\item Two new outer bounds on the trade-off of ER regenerating codes appear in \cite{Duursma}. These are obtained by extending the techniques of \cite{Tia} and \cite{SasSenKum}. The optimality of these bounds is not known, if we exclude the parameters $(n, k=3, d=n-1)$.
\item Constructions of ER regenerating codes which strictly improve upon the  space-sharing region of MBR and MSR codes appear in \cite{BirenPVK_canonical}, \cite{ernvall}, \cite{goparaju}. When $k = d = n-1$, the achievable regions presented in all these three works coincide (see Remark $1$, \cite{goparaju}).
\end{enumerate}

\subsection{Our results}

In this paper, we characterize the storage-repair-bandwidth trade-off of $(n, k=n-1, d=n-1), n \geq 5$ ER linear regenerating codes\footnote{The case  $n=4$ is already solved in \cite{Tia}, and the case $n < 4$ degenerates to trivial cases.}. This is done by deriving a new upper bound on the file size $B$ of ER linear regenerating codes for the case $k=d=n-1, n \geq 4$. The main result of this paper is stated below. 

\vspace{0.1in}

\begin{thm} \label{thm:new_bound_k_eq_d}
Consider an exact repair linear regenerating code, having parameters $(n, k = n-1, d = n-1), (\alpha, \beta), n \geq 4$. Then, the file size $B$ of the code is upper bounded by 
\begin{eqnarray} \label{eq:bound_rank_G}
B & \leq & \left \{ \begin{array}{c} \left \lfloor \frac{r(r-1)n\alpha + n(n-1)\beta}{r^2+r}\right \rfloor, \ \frac{d\beta}{r} \leq \alpha \leq \frac{d\beta}{r-1}, \\ 
 \hspace{1.5in} 2 \leq r \leq n - 2 \\  
(n-2)\alpha + \beta, \ \frac{d\beta}{n-1} \leq \alpha \leq \frac{d\beta}{n-2} \end{array} \right. .
\end{eqnarray}
\end{thm}

\vspace{0.1in}

The above theorem gives an upper bound on $B$ for the range of $\alpha$ given by $\beta \leq \alpha \leq (n-1)\beta$. Note that when $d=k=n-1$, $\alpha = \beta$ corresponds to the MSR point and $\alpha = (n-1)\beta$ corresponds to the MBR point. In Section \ref{sec:ach}, we will see that the outer bound on storage-repair-bandwidth trade-off corresponding to the bound in Theorem \ref{thm:new_bound_k_eq_d} coincides  with the achievability result provided in \cite{BirenPVK_canonical}. Thus, together with this achievability result, the new outer bound completely characterizes the trade-off of ER linear regenerating codes, for the case $k = d = n-1$, $n \geq 5$.

\vspace{0.1in}

\subsubsection{Illustration of Theorem \ref{thm:new_bound_k_eq_d} for the case of $(5, 4, 4)$ codes}

If we specialize \eqref{eq:bound_rank_G} for the case  $(5, 4, 4)$, we get 
\begin{eqnarray} \label{eq:bound_544_rankG_ours}
B & \leq & \left \{ \begin{array}{c} \left \lfloor \frac{5\alpha + 10\beta}{3} \right \rfloor, \ 2\beta \leq \alpha \leq 4\beta \\   
\left \lfloor \frac{15\alpha + 10\beta}{6} \right \rfloor, \ \frac{4\beta}{3} \leq \alpha \leq 2\beta \\  
3\alpha + \beta, \ \beta \leq \alpha \leq \frac{4\beta}{3} \end{array} \right. .
\end{eqnarray}
In Fig. \ref{fig:544_tradeoff}, we plot the  outer bound on the ``normalized" storage-repair-bandwidth trade-off between $\alpha/B$ and  $\beta/B$ corresponding to \eqref{eq:bound_544_rankG_ours}. The normalization is done with respect to the file size $B$. 
\begin{figure}[ht]
  \centering
\includegraphics[width=3in]{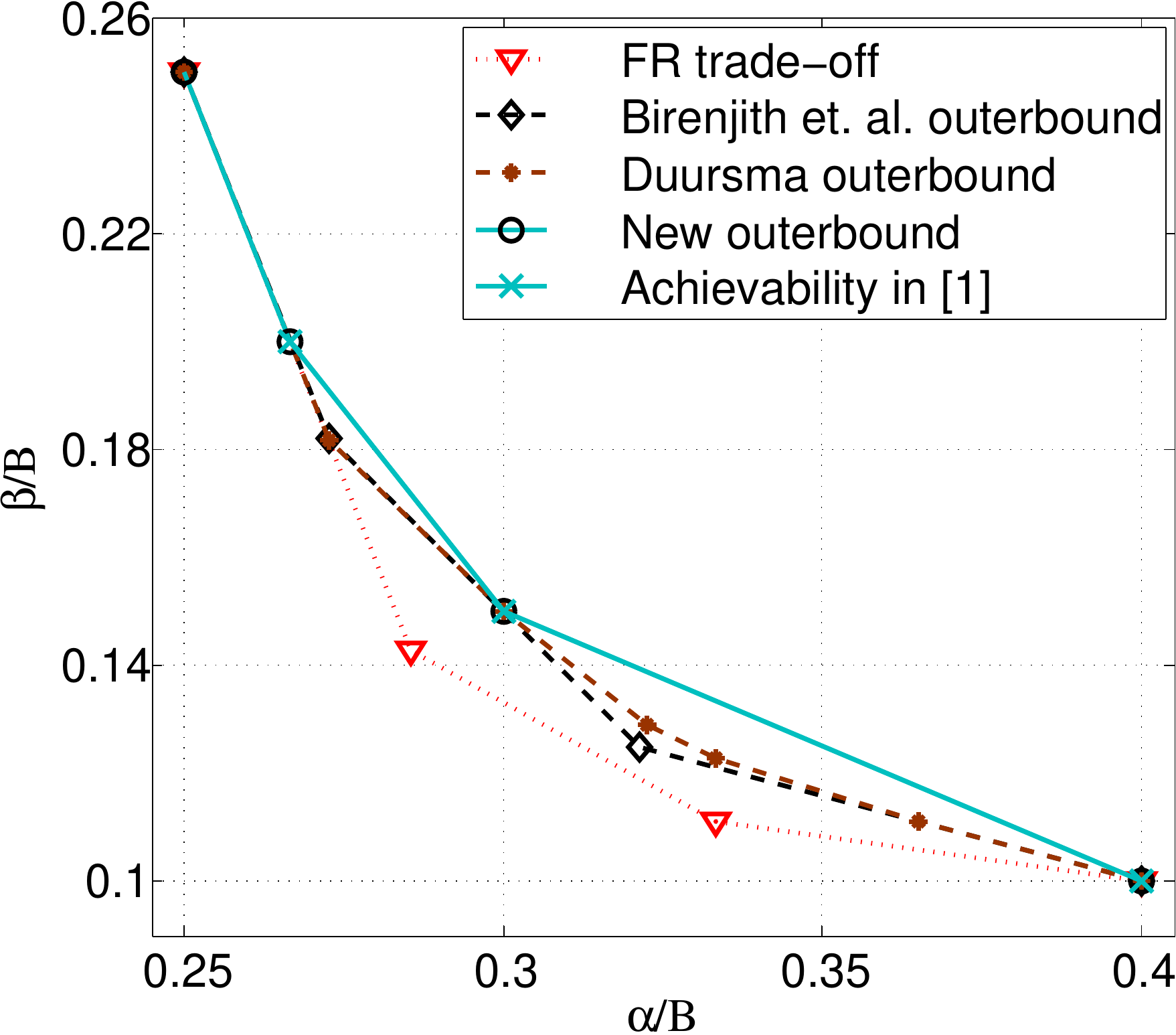}
  \caption{Comparison of outer bounds on the storage-repair-bandwidth trade-off of ER regenerating codes for the case  $(n = 5, k = d = 4)$. The new outer bound plotted in this figure is obtained under the assumption of linear regenerating codes.}
  \label{fig:544_tradeoff}
\end{figure}
In this figure, we have also plotted the following other curves:
\begin{enumerate}[1.]
\item The FR trade-off of $(5, 4, 4)$ regenerating codes.
\item The achievability result from \cite{BirenPVK_canonical} for the parameters $n = 5, d  = k = 4$. We see that our outer bound on the trade-off coincides with this achievability result, and thus establishes the optimality of the new outer bound.
\item The outer bounds on the trade-off obtained in \cite{SasSenKum} and \cite{Duursma} for  $(5, 4, 4)$ ER regenerating codes. We see that the new outer bound is tighter than both these other outer bounds, when the latter bounds are restricted to the case of linear regenerating codes. The expression for the file size bound appearing in \cite{SasSenKum}, when restricted to the case of linear codes, is given by (see Example $2$ of \cite{SasSenKum})
\begin{eqnarray} \label{eq:bound_544_rankG_Biren}
B & \leq & \left \{ \begin{array}{c} \left \lfloor \frac{7\alpha + 22\beta}{5} \right \rfloor, \ \frac{18}{7}\beta \leq \alpha \leq 4\beta \\  \\
\left \lfloor \frac{7\alpha + 6\beta}{3} \right \rfloor, \ \frac{3\beta}{2} \leq \alpha \leq \frac{18}{7}\beta \\  \\
3\alpha + \beta, \ \beta \leq \alpha \leq \frac{3\beta}{2} \end{array} \right.,
\end{eqnarray}
and the bound in \cite{Duursma}, when restricted to the case of linear codes, is given by (see Examples $4.3$ and $5.2$ of \cite{Duursma})
\begin{eqnarray} \label{eq:bound_544_rankG_Duursma}
B & \leq & \left \{ \begin{array}{c} \left \lfloor \frac{7\alpha + 22\beta}{5} \right \rfloor, \ \frac{23}{7}\beta \leq \alpha \leq 4\beta \\  \\
\left \lfloor \frac{21\alpha + 57\beta}{14} \right \rfloor, \ \frac{19}{7}\beta \leq \alpha \leq \frac{23}{7}\beta \\  \\
\left  \lfloor \frac{11\alpha + 19\beta}{6} \right \rfloor, \ \frac{5}{2}\beta \leq \alpha \leq \frac{19}{7}\beta \\  \\
\left  \lfloor \frac{13\alpha + 14\beta}{6} \right \rfloor, \ 2\beta \leq \alpha \leq \frac{5}{2}\beta \\ \\
\left  \lfloor \frac{7\alpha + 6\beta}{3} \right \rfloor, \ \frac{3}{2}\beta \leq \alpha \leq 2\beta \\ \\
3\alpha + \beta, \ \beta \leq \alpha \leq \frac{3}{2}\beta \end{array} \right..
\end{eqnarray}
\end{enumerate}

\subsection{Our Approach, and Some Preliminaries}

We make use of the fact that for the case of linear regenerating codes, maximizing the file size $B$ of the regenerating code is equivalent to minimizing the dimension of dual of the linear regenerating code. More formally, let $\mathcal{C}$ denote an $(n, k, d), (\alpha, \beta)$ ER linear regenerating code, having the (vector-symbol) alphabet $\mathbb{F}_q^{\alpha}$. Also, let the $B \times n\alpha$ matrix $G$ (whose entries are drawn from $\mathbb{F}_q$) denote a generator matrix for $\mathcal{C}$. To be precise, $G$ is the generator matrix for the underlying scalar code (say $\mathcal{C}_s$) of length $n\alpha$, where $\mathcal{C}_s$ is obtained by expanding each vector-symbol of $\mathcal{C}$ into $\alpha$ scalar symbols over $\mathbb{F}_q$.  Next, consider the $(n\alpha - B) \times n\alpha$ matrix $H$ which forms a parity check matrix of $\mathcal{C}_s$, i.e., $H$ generates the dual code of $\mathcal{C}_s$. In this paper, we will simply say that $H$ corresponds to the dual of the regenerating code $\mathcal{C}$, and also loosely identify the dual of the code $\mathcal{C}_s$ as the dual of the regenerating code $\mathcal{C}$ itself. Since the code is linear, we have $B = \text{rank}(G) = n\alpha - \text{rank}(H)$. Our approach in this paper will be to find a lower bound on $\text{rank}(H)$ and then convert it to an upper bound on $\text{rank}(G)$. 

Without loss of generality, we assume that the first $\alpha$ columns of $G$ generate the contents of the first node, the second $\alpha$ columns of $G$ generate the contents of the second node, and so on. The first $\alpha$ columns of  $H$ will together be referred to as the first thick column of $H$; similarly the second thick column and so on. For any set $S \subseteq [n]=\{1, 2, \ldots, n\}$, we will write $H|_S$ to denote restriction of $H$ to the thick columns indexed by the set $S$.

We will make use of the following properties of the matrix $H$ which were established in \cite{Duursma}. 

\vspace{0.1in}

\begin{lem}[Data Collection] \label{lem:H_datacollection}
$\text{Rank}\left({H|_S}\right) = (n-k)\alpha$, for any $S \subseteq [n]$ such that $|S| = n-k$. 
\end{lem}
\begin{proof}
This is a re-statement of Part $(1)$ of Proposition $2.1$ of \cite{Duursma}, and is equivalent to the data collection property.
\end{proof}

\vspace{0.1in}

\begin{lem}[Exact Repair] \label{lem:H_repair}
Assume that $d= n-1$. Then, under the assumption of ER, the row space of $H$ contains a collection of
$n\alpha$ vectors which can be arranged as the rows of an $n\alpha \times n\alpha$ matrix $H_{repair}$, as given below:
\begin{eqnarray} \label{eq:Hrepair}
H_{repair} & = & \left[ \begin{array}{c|c|c|c}  I_{\alpha} & A_{1,2} & & A_{1,n} \\
\hline \\
A_{2,1} & I_{\alpha}  & & A_{2, n} \\
\hline \\
& & \ \ \ddots \ \ & \\
\hline \\
A_{n,1} & A_{n,2} &  & I_{\alpha} \end{array}
\right],
\end{eqnarray}
where $I_{\alpha}$ denotes the identity matrix of size $\alpha$ and $A_{i,j}$ denotes an $\alpha \times \alpha$ matrix such that $\text{rank}\left(A_{i, j}\right) \leq \beta, 1 \leq i, j \leq n, i \neq j$.
\end{lem}
\begin{proof}
This follows from Part $(2)$ of Proposition $2.1$ of \cite{Duursma}, and is equivalent to the exact-repair property for the case $d = n-1$.
\end{proof}	

\vspace{0.1in}

\begin{note} \label{rem:dual_properties}
Note from Lemmas \ref{lem:H_datacollection} and \ref{lem:H_repair} that for the case of $d=k=n-1$, the matrix $H_{repair}$ by itself defines an $(n, k=n-1, d=n-1)(\alpha, \beta)$ regenerating code. Since $\text{rank}(H) \geq \text{rank}(H_{repair})$, we will assume that $H = H_{repair}$ while we derive a lower bound on the rank of $H$ for the case of $d=k=n-1$.
\end{note}

\vspace{0.1in}

The technical discussion appearing in the rest of the article is divided as follows:

\begin{enumerate}[1.]
\item In Section \ref{sec:ach}, we quickly review the achievability result from \cite{BirenPVK_canonical}, for the case of $k = d = n-1$. As mentioned before, the optimality of the new bound derived here will follow from this achievability result.
\item In Section \ref{sec:dimakis_via_dual}, we will re-derive  \eqref{eq:intro_cut_set_bd} for the case of ER linear codes, by calculating a simple lower bound on $\text{rank}(H)$.
\item In Sections \ref{sec:433} and \ref{sec:544}, we will refine the proof presented in Section \ref{sec:dimakis_via_dual}, and obtain the proof of Theorem \ref{thm:new_bound_k_eq_d} for the special cases of $(4, 3, 3)$ and $(5,4,4)$ ER linear regenerating codes, respectively. Our proofs for these two special cases will help us illustrate the key ideas that will be involved in the general proof. Note that the trade-off for the case of $(4,3,3)$ (including non-linear codes) has already been solved by Tian et. al. \cite{Tia}.
\item The proof for the general $(n, k = n - 1, d = n-1)$ will be subsequently presented in Section \ref{sec:general}. 
\end{enumerate}

\section{Achievable Region for $(n, k = n-1, d = n-1)$} \label{sec:ach}

In \cite{BirenPVK_canonical}, the authors give a construction of $(n, k=d, d)$ ER linear regenerating codes, and these are termed as canonical regenerating codes. When specialised to the case $d = n-1$, code constructions are obtained for the following points on the normalized storage vs repair-bandwidth plot:
\begin{equation} \label{eq:ach_Biren}
\left(\frac{\alpha}{B} , \frac{\beta}{B} \right)  =  \left(\frac{r}{n(r-1)}, \frac{r}{n(n-1)}\right), 2 \leq r \leq n-1.
\end{equation}
Note that in \eqref{eq:ach_Biren}, if we put $r = 2$, we get the MBR point, and as $r$ increases, points closer  to the MSR point are achieved. It is also proved that the point corresponding to $ r = n-1$ lies on the FR trade-off, on the line-segment whose one end point is the MSR point. An achievable region on the normalized storage vs repair-bandwidth plot, corresponding to \eqref{eq:ach_Biren} is obtained by 1) connecting the adjacent points in \eqref{eq:ach_Biren} by straight line-segments, and 2) drawing a line segment between the MSR point and the point corresponding to $r = n-1$. For example, if we set $n = 5$, the points of deflection on this achievable region are given by (see Fig. \ref{fig:544_tradeoff}.)
\begin{eqnarray}
\left(\frac{\alpha}{B} , \frac{\beta}{B} \right)  & = & \left(\frac{1}{4}, \frac{1}{4}\right)  , \ \text{MSR point} \\
\left(\frac{\alpha}{B} , \frac{\beta}{B} \right)  & = & \left(\frac{4}{15}, \frac{1}{5}\right)  , \ r = 4  \\
\left(\frac{\alpha}{B} , \frac{\beta}{B} \right)  & = & \left(\frac{3}{10}, \frac{3}{20}\right)  , \ r = 3  \\
\left(\frac{\alpha}{B} , \frac{\beta}{B} \right)  & = & \left(\frac{2}{5}, \frac{1}{10}\right)  , \ r = 2, \text{MBR point}.  
\end{eqnarray}

Now, to see that the outer bound on the normalized trade-off induced by the new file-size bound in Theorem \ref{thm:new_bound_k_eq_d} is the same as what is achieved by \cite{BirenPVK_canonical}, we note that 1) the equation of the line-segment obtained by connecting the two points $\left(\frac{r}{n(r-1)}, \frac{r}{n(n-1)}\right)$ and  $\left(\frac{(r+1)}{n((r+1)-1)}, \frac{(r+1)}{n(n-1)}\right)$, $ 2 \leq r \leq n-2$ is given by 
\begin{eqnarray}
r(r-1)n\left(\frac{\alpha}{B}\right) + n(n-1)\left(\frac{\beta}{B}\right) & = & r^2+r,
\end{eqnarray}
and 2) the equation of the line segment obtained by joining the MSR point and the point corresponding to $r=n-1$ is given by $(n-2)\left(\frac{\alpha}{B}\right) + \left(\frac{\beta}{B}\right) =  1$.

\section{A Derivation of \eqref{eq:intro_cut_set_bd} Based on Dual Code} \label{sec:dimakis_via_dual}

In this section, we will present a simple proof of \eqref{eq:intro_cut_set_bd} for ER linear regenerating codes. As we will see, our proof of Theorem \ref{thm:new_bound_k_eq_d}, to be presented later in this document, will be built up on proof of \eqref{eq:intro_cut_set_bd} that is presented here.

As before, we assume that $\mathcal{C}$ denotes an $(n, k, d = n-1) (\alpha, \beta)$ linear regenerating code, and the matrix $H$ generates the dual of $\mathcal{C}$. Also, recall that we use the notation $H|_S$ to denote the restriction of $H$ to the thick columns indexed by the set $S$, where $S \subseteq [n]$. The basic idea of the proof is to get a lower bound on the column rank of the matrix $H$. Towards this, define the quantities $\delta_j, 1 \leq j \leq n$ as follows:
\begin{eqnarray}
\delta_1 & = & \text{rank}(H|_{[1]}), \label{eq:proof_diamkis_aa}\\
\delta_j & = & \text{rank}(H|_{[j]}) - \text{rank}(H|_{[j-1]}), 2 \leq j \leq n, \label{eq:proof_diamkis_0}
\end{eqnarray}
where we have used the notation $[t] = \{1, 2, \ldots, t\}$ for any positive integer $t$.

We claim that
\begin{eqnarray} \label{eq:proof_diamkis_1}
\delta_{j} = \alpha, \ 1 \leq j \leq n-k, 
\end{eqnarray}
and 
\begin{eqnarray} \label{eq:proof_diamkis_2}
\delta_{j} \geq (\alpha - (j-1)\beta)^+, \ n-k + 1 \leq j \leq n, 
\end{eqnarray}
where the quantity $a^+$ denotes $\max(a, 0)$. Here, \eqref{eq:proof_diamkis_1} follows because we know from Lemma \ref{lem:H_datacollection} that any $n-k$ thick  columns of $H$ has rank given by $(n-k)\alpha$. To see why \eqref{eq:proof_diamkis_2} is true, focus on the $j^{\text{th}}$ thick row of $H_{repair}$ (i.e., the rows from $(j-1)\alpha + 1$ to $j\alpha$ of $H_{repair}$) and note that 
\begin{eqnarray}
\delta_j & \geq & \left(\text{rank}(I_{\alpha}) - \sum_{\ell=1}^{j-1}\text{rank}(A_{j,\ell})\right)^+  \label{eq:dimakis_proof_3a} \\
& \geq & \left(\alpha - (j-1)\beta \right)^+,  \ n-k + 1 \leq j \leq n, \label{eq:dimakis_proof_3}
\end{eqnarray}
where \eqref{eq:dimakis_proof_3} follows because, we know from Lemma \ref{lem:H_repair} that $\text{rank}(A_{i,j}) \leq \beta$. 

Now, the (column) rank of the matrix $H$ can be lower bounded as 
\begin{eqnarray}
\text{rank}(H) & = & \sum_{i=j}^{n}\delta_j \\
& \geq & (n-k)\alpha + \sum_{j=n-k+1}^{n}\left(\alpha - (j-1)\beta \right)^+, \label{eq:dimakis_proof_4}
\end{eqnarray}
where \eqref{eq:dimakis_proof_4} follows from \eqref{eq:proof_diamkis_1} and \eqref{eq:proof_diamkis_2}. From this, it follows that the file size $B$ of the code $\mathcal{C}$ can be upper bounded as 
\begin{eqnarray}
B &  =  & n\alpha - \text{rank}(H) \\
& \leq & n\alpha - (n-k)\alpha  - \sum_{j=n-k+1}^{n}\left(\alpha - (j-1)\beta \right)^+ \\
& = & \sum_{j=n-k+1}^{n}\left( \alpha - \left(\alpha - (j-1)\beta \right)^+ \right) \\
& = & \sum_{j=n-k+1}^{n}\min(\alpha, (j-1)\beta) \label{eq:dimakis_proof_5} \\
& = & \sum_{j=0}^{k-1}\min(\alpha, (d-j)\beta), \label{eq:dimakis_proof_6}
\end{eqnarray}
where \eqref{eq:dimakis_proof_5} follows from noting that $\alpha - \left(\alpha - (j-1)\beta \right)^+ = \min(\alpha, (j-1)\beta)$ and \eqref{eq:dimakis_proof_6} follows from our assumption that $d = n-1$.

\section{The Trade-off of $(4, 3, 3)$  ER Linear Regenerating Codes Based on Dual Code} \label{sec:433}

In this section, we will re-derive the trade-off of $(4, 3, 3)$ linear regenerating codes, which was originally obtained by \cite{Tia}. Our proof here will be built on the proof of $(1)$ which was presented in Section \ref{sec:dimakis_via_dual}. We will prove that, when restricted to the case of ER, it is possible to get a lower bound on $\text{rank}(H)$ that is in general tighter than what is given in \eqref{eq:dimakis_proof_4}. The following theorem specialises Theorem \ref{thm:new_bound_k_eq_d} for the case of $(4, 3, 3)$ and states the result (to be proved here) in terms of $\text{rank}(H)$.

\vspace{0.1in}

\begin{thm} \label{thm:433}
Consider an exact repair linear regenerating code $\mathcal{C}$, having parameters $(n=4, k=3, d=3,), (\alpha, \beta)$. Let the matrix $H$ correspond to the dual of the code $\mathcal{C}$. Then, the rank of the matrix $H$ is lower bounded by
\begin{eqnarray} \label{eq:bound_rank_H_433}
\text{rank}(H) & \geq & \left \{ \begin{array}{c} \left \lceil \frac{8\alpha - 6\beta}{3} \right \rceil, \ 1.5\beta \leq \alpha \leq 3\beta \\  \\
2\alpha - \beta, \ \beta \leq \alpha \leq 1.5\beta \end{array} \right. .
\end{eqnarray}
\end{thm}

\vspace{0.1in}

Before we prove Theorem \ref{thm:433},we note the following points regarding this theorem. 
\begin{enumerate}[1.]
\item For the case of $(4, 3, 3)$, $\alpha = \beta$ corresponds to the MSR point and $\alpha = 3\beta$ corresponds to the MBR point.
\item From \eqref{eq:dimakis_proof_4}, we see that 
\begin{eqnarray} \label{eq:433_proof_1}
\text{rank}(H) & \geq & 2\alpha - \beta, \ \beta \leq \alpha \leq 2\beta. \label{eq:proof_433_en}
\end{eqnarray}
Thus, to prove Theorem \ref{thm:433}, we only need to prove that 
\begin{eqnarray} 
\text{rank}(H) & \geq & \left \lceil \frac{8\alpha - 6\beta}{3} \right \rceil, \ 1.5\beta \leq \alpha \leq 3\beta. 
\end{eqnarray}
In fact, we will simply prove that 
\begin{eqnarray} \label{eq:433_proof_4}
\text{rank}(H)  \geq  \left \lceil \frac{8\alpha - 6\beta}{3} \right \rceil,
\end{eqnarray} 
without bothering about the range of $\alpha$. Note that given \eqref{eq:proof_433_en}, this suffices to prove Theorem \ref{thm:433}.
\item To see that the bound in Theorem \ref{thm:433} is tighter than the bound in \eqref{eq:dimakis_proof_4}, note that
\begin{eqnarray} \label{eq:433_proof_2}
\left \lceil \frac{8\alpha - 6\beta}{3} \right \rceil & > & 2\alpha - \beta,  \ 1.5\beta < \alpha  \leq 2\beta
\end{eqnarray}
and 
\begin{eqnarray} \label{eq:433_proof_3}
\left \lceil \frac{8\alpha - 6\beta}{3} \right \rceil & > & 3\alpha - 3\beta,  \ 2\beta \leq \alpha  < 3\beta,
\end{eqnarray}
where $2\alpha - \beta$ and  $3\alpha - 3\beta$ respectively denote the bounds obtained in \eqref{eq:dimakis_proof_4} for the cases when $\beta \leq \alpha \leq 2\beta$ and $2\beta  \leq \alpha \leq 3\beta$. 
\end{enumerate}

\vspace{0.1in}

A comparison of the bounds in  \eqref{eq:dimakis_proof_4} and \eqref{eq:bound_rank_H_433} is shown in Fig. \ref{fig:433_rankH_comparison}. Here, we plot the two lower bounds on $\text{rank}(H)$ as a function of $\alpha$, when $\beta$ is fixed as $48$.
\begin{figure}[h]
  \centering
\includegraphics[width=3in]{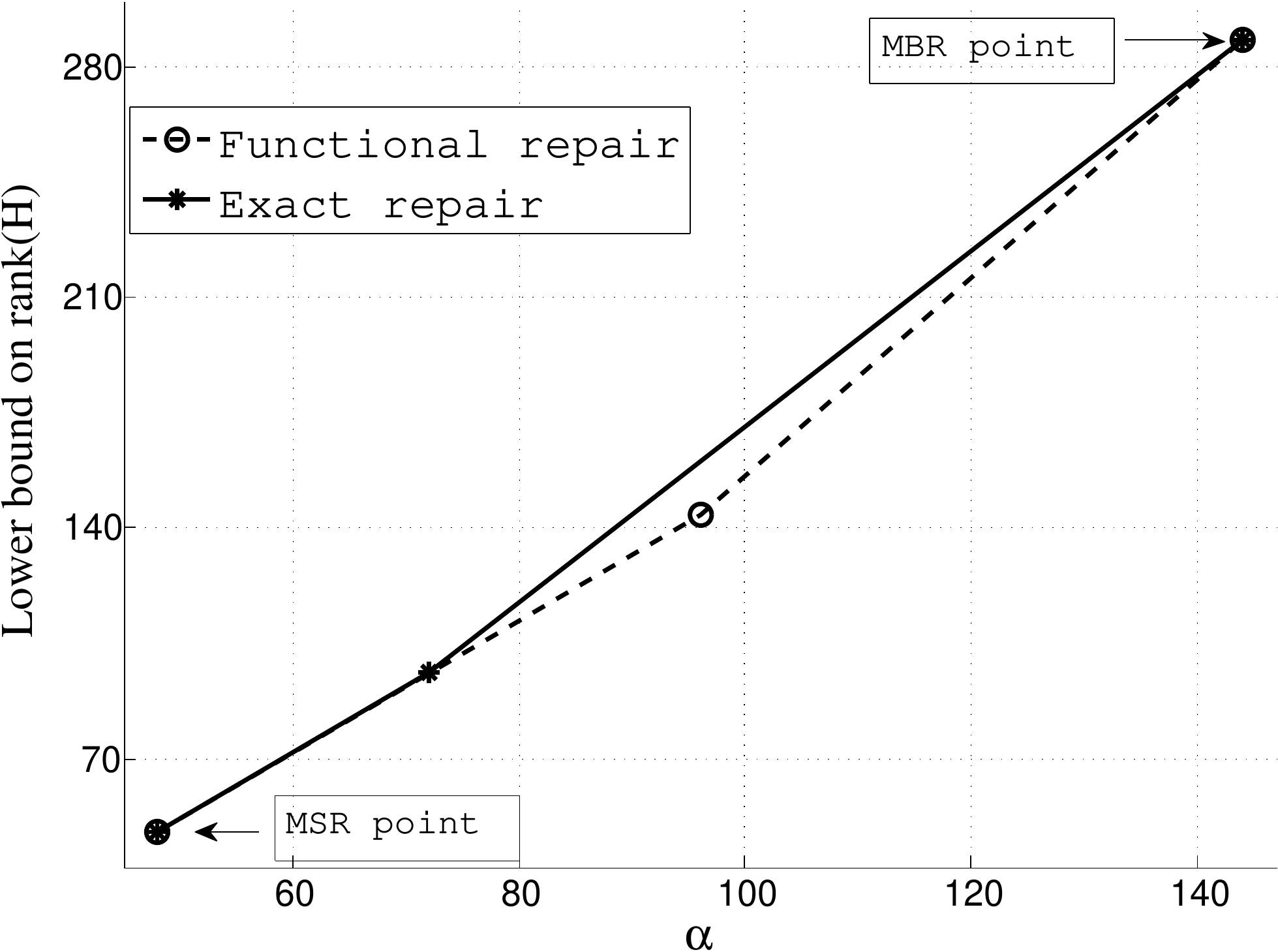}
  \caption{Comparison of the lower bounds on $\text{rank}(H)$ as function of $\alpha$, for the case of $(n = 4, k = 3, d = 3)$ with $\beta = 48$. The dashed and the solid lines correspond to the cases of functional and exact repairs, respectively. See  \eqref{eq:dimakis_proof_4} and \eqref{eq:bound_rank_H_433} for the corresponding equations.}
  \label{fig:433_rankH_comparison}
\end{figure}

\vspace{0.1in}

\subsection{Proof of Theorem \ref{thm:433}} \label{sec:proof_thm_433}

We begin with certain notation needed to prove Theorem \ref{thm:433}. For any matrix $B$ (over $\mathbb{F}_q$), we will use $\mathcal{S}(B)$ to denote the column space of $B$. We will write $\rho(B)$ to mean $\text{rank}(B)$, which is also the same as the dimension of the space $\mathcal{S}(B)$. Next, define $H^{(4)} = H_{repair}$. Also let the matrix $H^{(4)}_j$ denote the $j^{\text{th}}$ thick column of $H^{(4)}, 1 \leq j \leq 4$, i.e., $H^{(4)} = [H^{(4)}_1 \ H^{(4)}_2  H^{(4)}_3 \ H^{(4)}_4 ]$. Next, define the matrices $H^{(3)}_j,  2 \leq j \leq 4$, such that the columns of $H^{(3)}_j$ form a basis for the vector space $\mathcal{S}\left(H^{(4)}_j\right) \cap \mathcal{S}\left(H^{(4)}|_{[j-1]}\right)$. Also define the matrix $H^{(3)}$ as
\begin{eqnarray}
H^{(3)} = [H^{(3)}_2  H^{(3)}_3 \ H^{(3)}_4 ], 
\end{eqnarray}
i.e., $H^{(3)}$ is obtained by stacking the columns of $H^{(3)}_j, 2 \leq j \leq 4$. Notice that the first thick column of $H^{(3)}$ is denoted as $H^{(3)}_2$ instead of $H^{(3)}_1$ (and so on). This has been done intentionally for notational convenience.

The basic idea of our proof for the case of $(4, 3, 3)$ comes from the observation that $\rho(H^{(4)}) \geq \rho(H^{(3)})$. We will show here that \eqref{eq:433_proof_4} is a necessary condition for this to be true. Towards this, we will firstly establish some more notation needed for the proof. We will then separately compute (or bound) the ranks of the two matrices $H^{(4)}$ and $H^{(3)}$. Finally, we will show that the comparison of the two ranks yields the desired bound.

\vspace{0.1in}

\subsubsection{Some Additional Notation }

For any two subspaces $W_1$ and $W_2$, we write $W_1 \subseteq W_2$ to mean that $W_1$ is a subspace of $W_2$. Equation \eqref{eq:Hrepair} will be denoted as $H^{(4)}  = \left( A^{(4)}_{i,j} , 1 \leq i, j \leq 4 \ \right )$, where $A^{(4)}_{i,i} = I_{\alpha}, 1 \leq i \leq 4$, and when $i \neq j$, we have added a superscript on $A_{i,j}$. Note that in this notation, the $j^{\text{th}}$ thick column of $H^{(4)}$ is given by $H^{(4)}_j = \left( A^{(4)}_{i,j} , 1 \leq  i \leq 4 \ \right )$. Also note that 
\begin{eqnarray} \label{eq:433_proof_0a}
\rho\left(H^{(4)}_j\right) & = & \rho\left( A^{(4)}_{j,j}\right),  \ 1 \leq j \leq 4.
\end{eqnarray}

In terms of block sub-matrices, the matrix $H^{(3)}$ will be identified as $H^{(3)}  = \left( A^{(3)}_{i,j} , 1 \leq i \leq 4, 2 \leq j \leq 4 \ \right )$, where  $A^{(3)}_{i,j}$ is an $\alpha \times \rho(H^{(3)}_j)$ matrix (over $\mathbb{F}_q$) such that 
\begin{eqnarray}
\mathcal{S}\left(A^{(3)}_{i,j} \right) & \subseteq & \mathcal{S}\left(A^{(4)}_{i,j} \right) \ \bigcap \ \sum_{\ell=1}^{j-1}\mathcal{S}\left(A^{(4)}_{i,\ell}\right). \label{eq:433_proof_0b}
\end{eqnarray}
Note that \eqref{eq:433_proof_0b} is a direct consequence of the definition of the matrix $H^{(3)}$. Here, we would like to clarify that \eqref{eq:433_proof_0b} is not equivalent to the definition of $H^{(3)}$. The definition of $H^{(3)}$ demands additional restrictions on the matrices $\{A_{i,j}^{(3)}\}$ so that the columns of $H^{(3)}_j$ form a basis for the vector space $\mathcal{S}\left(H^{(4)}_j\right) \cap \mathcal{S}\left(H^{(4)}|_{[j-1]}\right), 2 \leq j \leq 4$.

Next, observe that the $j^{\text{th}}$ thick column of $H^{(3)}$ can be written in terms of the block sub-matrices as $H^{(3)}_j = \left( A^{(3)}_{i,j}, 1 \leq  i \leq 4 \ \right), 2 \leq j \leq 4$. Also, note that 
\begin{eqnarray}
\rho\left(H^{(3)}_j\right) & = & \rho\left( A^{(3)}_{j,j}\right), \ 2 \leq j \leq 4.\label{eq:433_proof_0}
\end{eqnarray}
To see why \eqref{eq:433_proof_0} is true, firstly observe that the $\rho\left(H^{(3)}_j\right)$ columns of $H^{(3)}_j$ can be extended to a basis for $\mathcal{S}\left(H^{(4)}_j\right)$ by adding exactly $\rho\left(H^{(4)}_j\right) - \rho\left(H^{(3)}_j\right)$ additional columns. This implies that the $\rho\left(H^{(3)}_j\right)$ columns of $A^{(3)}_{j,j}$ can be extended to a basis for $\mathcal{S}\left(A^{(4)}_{j,j}\right)$ by adding at most $\rho\left(H^{(4)}_j\right) - \rho\left(H^{(3)}_j\right)$ additional columns. But then, we know from \eqref{eq:433_proof_0a} that $\rho\left(H^{(4)}_j\right)  =  \rho\left( A^{(4)}_{j,j}\right)$. Hence, it must indeed be true that $\rho\left(H^{(3)}_j\right) =  \rho\left( A^{(3)}_{j,j}\right), \ 2 \leq j \leq 4$.

\vspace{0.1in}

\subsubsection{$\text{Rank}(H^{(4)})$}

Let us define the quantities $\delta_j, 1 \leq j \leq 4$ in the same manner as we did in \eqref{eq:proof_diamkis_aa} and \eqref{eq:proof_diamkis_0} for the proof of the FR-trade-off, i.e., 
\begin{eqnarray}
\delta_1 & = & \rho\left(H^{(4)}_1\right) \ = \ \rho\left(A^{(4)}_{1,1}\right), \label{eq:433_proof_5z}\\
\delta_j & = & \rho\left(H^{(4)}|_{[j]}\right) - \rho\left(H^{(4)}|_{[j-1]}\right), \ 2 \leq j \leq 4, \label{eq:433_proof_5}.
\end{eqnarray}
From the discussion in Section \ref{sec:dimakis_via_dual} (see \eqref{eq:dimakis_proof_3a}), we know that 
\begin{eqnarray}
\delta_j & \geq & \left(\rho\left(A_{j,j}^{(4)}\right) - \sum_{\ell=1}^{j-1}\rho\left(A_{j,\ell}^{(4)}\right)\right)^+, \ 2 \leq j \leq 4. \label{eq:433_proof_5a}
\end{eqnarray}

Thus, let us assume that 
\begin{eqnarray}
\delta_j & = & \left(\rho\left(A_{j,j}^{(4)}\right) - \sum_{\ell=1}^{j-1}\rho\left(A_{j,\ell}^{(4)}\right)\right)^+  +   \alpha_j, \ 2 \leq j \leq 4 \label{eq:433_proof_6},
\end{eqnarray}
where $\{\alpha_j, 2 \leq j \leq 4\}$ are non-negative integers. The rank of the matrix $H^{(4)}$ can now be written as
\begin{eqnarray}
\rho\left( H^{(4)} \right) & = & \sum_{j=1}^{4}\delta_j \\
& = & \rho\left(A^{(4)}_{1,1}\right) \ + \  \sum_{j=2}^{4} \left\{ \left(\rho\left(A_{j,j}^{(4)}\right) - \sum_{\ell=1}^{j-1}\rho\left(A_{j,\ell}^{(4)}\right)\right)^+  +   \alpha_j \right\} \label{eq:433_proof_7}.
\end{eqnarray}

\vspace{0.1in}

\begin{note}
Note that in \eqref{eq:433_proof_5a}, we could have written the bound on $\delta_j$ directly in terms of $\alpha, \beta$, like we did in \eqref{eq:dimakis_proof_3}. The reason for not doing so at this stage of the proof, and keeping the bound on $\delta_j$ only in terms of ranks of $\{A_{i,j}^{(4)}\}$ will become evident when we discuss the case of $(5, 4, 4)$ regenerating codes in Section \ref{sec:544}. See Remark \ref{rem:resuse433} as well.
\end{note}

\vspace{0.1in}

\subsubsection{$\text{Rank}(H^{(3)})$}

The following lemma obtains a relation between the ranks of the matrices $\{A^{(3)}_{i,j}\}$ and the ranks of the matrices $\{A^{(4)}_{i,j}\}$. This result will be used to obtain a lower bound on the rank of $H^{(3)}$.

\vspace{0.1in}

\begin{lem} \label{lem:433_intersections}
\begin{enumerate}[a)]
\item 
\begin{eqnarray}
\rho\left(A^{(3)}_{j,j} \right) & = & \rho\left(A^{(4)}_{j,j} \right) \ - \ \delta_j \label{eq:433_proof_8} \\
& = & \rho\left(A^{(4)}_{j,j} \right) \ - \ \left\{ \left(\rho\left(A_{j,j}^{(4)}\right) - \sum_{\ell=1}^{j-1}\rho\left(A_{j,\ell}^{(4)}\right)\right)^+  +   \alpha_j\right\}, \ 2 \leq j \leq 4. \label{eq:433_proof_9}
\end{eqnarray}
\item   
\begin{eqnarray}
\sum_{\ell=2}^{j-1} \rho\left(A^{(3)}_{j,\ell} \right) & \leq & \sum_{\ell=1}^{j-1}\rho\left(A^{(4)}_{j,\ell} \right) \ - \ \rho\left(A^{(3)}_{j,j} \right), \ 3 \leq j \leq 4. \label{eq:433_proof_10}
\end{eqnarray}
\end{enumerate}
\end{lem}
\begin{proof}
Let us prove \eqref{eq:433_proof_8} first. Using the definition of $\delta_j$ from \eqref{eq:433_proof_5}, we get
\begin{eqnarray}
\delta_j & = & \rho\left(H^{(4)}|_{[j]}\right) - \rho\left(H^{(4)}|_{[j-1]}\right) \\
& = & \text{dim}\left(\mathcal{S}\left(H^{(4)}|_{[j-1]} \right) + \mathcal{S}\left(H^{(4)}_j \right)\right) -  \text{dim}\left(\mathcal{S}\left(H^{(4)}|_{[j-1]}\right)\right) \label{eq:433_proof_11} \\ 
& = & \text{dim}\left(\mathcal{S}\left(H^{(4)}_j \right)\right) - \text{dim}\left(\mathcal{S}\left(H^{(4)}|_{[j-1]} \right) \cap \mathcal{S}\left(H^{(4)}_j \right)\right) \label{eq:433_proof_12} \\
& = & \rho\left(H^{(4)}_j \right) - \rho\left(H^{(3)}_j \right) \label{eq:433_proof_13} \\
& = & \rho\left(A^{(4)}_{j,j} \right) - \rho\left(A^{(3)}_{j,j} \right), \label{eq:433_proof_14}
\end{eqnarray} 
where in \eqref{eq:433_proof_12} we have used the fact that for any two subspaces $W_1, W_2$, $\text{dim}(W_1 + W_2)  = \text{dim}(W_1) + \text{dim}(W_2) - \text{dim}(W_1 \cap W_2)$. Equation  \eqref{eq:433_proof_13} follows from the definition of $H^{(3)}_j$, while \eqref{eq:433_proof_14} follows from \eqref{eq:433_proof_0a} and \eqref{eq:433_proof_0}. This completes the proof of \eqref{eq:433_proof_8}. Equation \eqref{eq:433_proof_9} now follows directly from \eqref{eq:433_proof_6}.

We will next prove \eqref{eq:433_proof_10} which is the second claim in the lemma. Towards this, observe from \eqref{eq:433_proof_0b} that we have $\mathcal{S}\left(A^{(3)}_{j,j}\right) \subseteq  \sum_{\ell = 1}^{j-1}\mathcal{S}\left(A^{(4)}_{j,\ell}\right)$ and thus, we get that
\begin{eqnarray}
\rho \left(A^{(3)}_{j,j}\right) & \leq &  \text{dim} \left( \sum_{\ell = 1}^{j-1}\mathcal{S}\left(A^{(4)}_{j,\ell}\right) \right). \label{eq:433_proof_15}
\end{eqnarray}
The right hand side of \eqref{eq:433_proof_15} can be upper bounded as follows:
\begin{eqnarray}
\text{dim} \left( \sum_{\ell = 1}^{j-1}\mathcal{S}\left(A^{(4)}_{j,\ell}\right) \right) & = & \text{dim} \left( \sum_{\ell = 1}^{j-2}\mathcal{S}\left(A^{(4)}_{j,\ell}\right)  + \mathcal{S}\left(A^{(4)}_{j,j-1}\right)\right) \label{eq:433_proof_16} \\
& = & \text{dim} \left( \sum_{\ell = 1}^{j-2}\mathcal{S}\left(A^{(4)}_{j,\ell}\right) \right) + \text{dim} \left(  \mathcal{S}\left(A^{(4)}_{j,j-1}\right)\right) - \text{dim} \left( \sum_{\ell = 1}^{j-2}\mathcal{S}\left(A^{(4)}_{j,\ell}\right)  \cap \mathcal{S}\left(A^{(4)}_{j,j-1}\right)\right) \label{eq:433_proof_17} \\
& \leq & \text{dim} \left( \sum_{\ell = 1}^{j-2}\mathcal{S}\left(A^{(4)}_{j,\ell}\right) \right) + \text{dim} \left(  \mathcal{S}\left(A^{(4)}_{j,j-1}\right)\right) - \text{dim} \left( \mathcal{S}\left(A^{(3)}_{j,j-1}\right)\right), \label{eq:433_proof_18} \\
& = & \text{dim} \left( \sum_{\ell = 1}^{j-2}\mathcal{S}\left(A^{(4)}_{j,\ell}\right) \right) + \rho\left(A^{(4)}_{j,j-1}\right) - \rho\left(A^{(3)}_{j,j-1}\right), \label{eq:433_proof_19}
\end{eqnarray}
where \eqref{eq:433_proof_18} follows from \eqref{eq:433_proof_0b}. The term $\text{dim} \left( \sum_{\ell = 1}^{j-2}\mathcal{S}\left(A^{(4)}_{j,\ell}\right) \right)$ appearing in \eqref{eq:433_proof_19} can be further upper bounded (for the case when $j=4$. If $j=3$ \eqref{eq:433_proof_19} completes the proof) by following a similar sequence of steps as in \eqref{eq:433_proof_16} - \eqref{eq:433_proof_19}. Combining with \eqref{eq:433_proof_15}, we eventually get that 
\begin{eqnarray}
\rho \left(A^{(3)}_{j,j}\right) & \leq &   \sum_{\ell=1}^{j-1}\rho\left(A^{(4)}_{j,\ell} \right) \ - \ \sum_{\ell=2}^{j-1} \rho\left(A^{(3)}_{j,\ell} \right), \ 3 \leq j \leq 4.  
\end{eqnarray}
This completes the proof of \eqref{eq:433_proof_10}.
\end{proof}

\vspace{0.1in}

We will now use the result in Lemma \ref{lem:433_intersections} to get a lower bound on the rank of $H^{(3)}$. The steps that we follow here are similar to those appearing in the calculation of the rank of $H^{(4)}$. Thus, let us define the quantities $\gamma_j, 2 \leq j \leq 4$ such that
\begin{eqnarray}
\gamma_2 & = & \rho\left(H^{(3)}_2\right), \\
\gamma_j & = & \rho\left(H^{(3)}|_{\{2, \ldots, j \}}\right) - \rho\left(H^{(3)}|_{ \{2, \ldots, j-1 \} }\right), \ 3 \leq j \leq 4. \label{eq:433_proof_11r}
\end{eqnarray}
The quantity $\gamma_2$ is given by 
\begin{eqnarray}
\gamma_2 & = & \rho\left(H^{(3)}_2\right) \\
& = & \rho\left(A^{(3)}_{2,2}\right) \label{eq:433_proof_11aa} \\
& = & \rho\left(A^{(4)}_{2,2} \right) \ - \ \left\{ \left(\rho\left(A_{2,2}^{(4)}\right) - \rho\left(A_{2,1}^{(4)}\right)\right)^+  +   \alpha_2\right\} \label{eq:433_proof_11a},
\end{eqnarray}
where \eqref{eq:433_proof_11aa} and \eqref{eq:433_proof_11a} follow from \eqref{eq:433_proof_0} and \eqref{eq:433_proof_9}, respectively. The quantities $\gamma_j, 3 \leq j \leq 4 $ can be lower bounded in the same way we lower bounded $\delta_j, 2 \leq j \leq 4$ in \eqref{eq:433_proof_5a}. Thus, we get that
\begin{eqnarray}
\gamma_j & \geq & \left(\rho\left(A_{j,j}^{(3)}\right) - \sum_{\ell=2}^{j-1}\rho\left(A_{j,\ell}^{(3)}\right)\right)^+ \\
& \geq & \rho\left(A_{j,j}^{(3)}\right) - \sum_{\ell=2}^{j-1}\rho\left(A_{j,\ell}^{(3)}\right) \\
& \geq & 2\rho\left(A_{j,j}^{(3)}\right) - \sum_{\ell=1}^{j-1}\rho\left(A_{j,\ell}^{(4)}\right) \label{eq:433_proof_12r}\\
& = & 2\left[ \rho\left(A^{(4)}_{j,j} \right) \ - \ \left\{ \left(\rho\left(A_{j,j}^{(4)}\right) - \sum_{\ell=1}^{j-1}\rho\left(A_{j,\ell}^{(4)}\right)\right)^+  +   \alpha_j\right\}\right] - \sum_{\ell=1}^{j-1}\rho\left(A_{j,\ell}^{(4)}\right), \ 3 \leq j \leq 4, \label{eq:433_proof_13r}
\end{eqnarray}
where \eqref{eq:433_proof_12r} and \eqref{eq:433_proof_13r} follow from \eqref{eq:433_proof_10} and \eqref{eq:433_proof_9}, respectively. The rank of the matrix $H^{(3)}$ is now given by 
\begin{eqnarray}
\rho\left(H^{(3)}\right) & = &\sum_{j=2}^{4}\gamma_j, \label{eq:proof_433_rankH3_short}
\end{eqnarray}
where $\gamma_2$ is given by \eqref{eq:433_proof_11a}, and $\gamma_3, \gamma_4$ are lower bounded as given by \eqref{eq:433_proof_13r}. 

\vspace{0.1in}

\subsubsection{Comparison of the ranks of the matrices $H^{(3)}$ and $H^{(4)}$}

We are now in a position to compare the ranks of the matrices $H^{(3)}$ and $H^{(4)}$. Recall from \eqref{eq:433_proof_7} that the rank of the matrix $H^{(4)}$ is given by
\begin{eqnarray}
\rho\left( H^{(4)} \right) & = & \rho\left(A^{(4)}_{1,1}\right) \ + \  \sum_{j=2}^{4} \left\{ \left(\rho\left(A_{j,j}^{(4)}\right) - \sum_{\ell=1}^{j-1}\rho\left(A_{j,\ell}^{(4)}\right)\right)^+  +   \alpha_j \right\}. \label{eq:433_proof_14r}
\end{eqnarray}
The goal is to obtain a lower bound on $\sum_{j=2}^{4}\alpha_j$ via the comparison $\rho\left( H^{(4)} \right) \geq \rho\left( H^{(3)} \right)$, and then use this lower bound in \eqref{eq:433_proof_14r} to get the desired lower bound on $\rho\left( H^{(4)} \right)$. One can verify that the comparison yields the following lower bound on $\sum_{j=2}^{4}\alpha_j$:
\begin{eqnarray}
\sum_{j=2}^{4}\alpha_j & \geq & \frac{1}{3}\left\{ -\rho\left(A^{(4)}_{1,1}\right) + \rho\left(A^{(4)}_{2,2}\right)  +
2\sum_{j=3}^{4}\rho\left(A^{(4)}_{j,j}\right) - \right. \nonumber \\ 
& & \left. \left[2\left( \rho\left(A^{(4)}_{2,2}\right) - \rho\left(A^{(4)}_{2,1}\right)\right)^+     + 3\sum_{j=3}^{4}\left( \rho\left(A^{(4)}_{j,j}\right) - \sum_{\ell=1}^{j-1}\rho\left(A^{(4)}_{j,\ell}\right)\right)^+ + \sum_{j=3}^{4} \sum_{\ell=1}^{j-1}\rho\left(A^{(4)}_{j,\ell}\right)\right]
 \right\} \label{eq:433_proof_15r}
\end{eqnarray}
An upper bound on the rank of $H^{(4)}$ is now obtained by substituting  \eqref{eq:433_proof_15r} back in \eqref{eq:433_proof_14r}. The result is stated formally in the following theorem: 

\vspace{0.1in}

\begin{thm} \label{thm:433_rankH4}
\begin{eqnarray}
\rho\left( H^{(4)} \right) & \geq & \frac{1}{3}\left\{ 2\sum_{j=1}^{4}\rho\left(A^{(4)}_{j,j}\right) - \sum_{j=2}^{4} \sum_{\ell=1}^{j-1}\rho\left(A^{(4)}_{j,\ell}\right)\right\}. \label{eq:433_proof_16r}
\end{eqnarray}
\end{thm}

Finally, to get the bound in Theorem \ref{thm:433}, we invoke the facts that $\rho\left(A^{(4)}_{j,j}\right) = \alpha,  1 \leq j \leq 4$ and $\rho\left(A^{(4)}_{i,j}\right) \leq \beta, \ 1 \leq i, j \leq 4, i \neq j$. Using these expressions in \eqref{eq:433_proof_16r}, we get that
\begin{eqnarray}
\rho\left( H^{(4)} \right) & \geq & \frac{1}{3}\left\{ 2\sum_{j=1}^{4}\alpha - \sum_{j=2}^{4} (j-1)\beta\right\} \\
& = &  \frac{8\alpha - 6\beta}{3}. 
\end{eqnarray}
The use of the ceil function in Theorem \ref{thm:433} is justified by the fact the ranks are integers.

\vspace{0.1in}

\begin{note} \label{rem:resuse433}
Note that in the preceding discussion, we never used the facts that $\rho\left(A^{(4)}_{j,j}\right) = \alpha,  1 \leq j \leq 4$ and $\rho\left(A^{(4)}_{i,j}\right) \leq \beta, \ 1 \leq i, j \leq 4, i \neq j$ till (including) Theorem \ref{thm:433_rankH4}. The lower bound in Theorem \ref{thm:433_rankH4} holds for any matrix $H^{(4)}  = [H^{(4)}_1 \ H^{(4)}_2 \ H^{(4)}_3 \ H^{(4)}_4 ] = \left( A^{(4)}_{i,j} , 1 \leq i, j \leq 4 \ \right )$ having the following properties:
\begin{enumerate}[1.]
\item For any $j, 1 \leq j \leq 4$, the columns of $H^{(4)}_j$ are all linearly independent,
\item $\rho\left( H^{(4)}_j \right) = \rho\left( A^{(4)}_{j,j} \right), \ 1 \leq j \leq 4$.
\end{enumerate}
As we will see, this fact will enable us to reuse Theorem \ref{thm:433_rankH4} in a certain way, for the case of $(5, 4, 4)$ as well.
 
\end{note}

\section{The Trade-off of $(5, 4, 4)$ ER Linear Regenerating Codes} \label{sec:544}

In this section, we will prove Theorem \ref{thm:new_bound_k_eq_d} for the case of $(5, 4, 4) (\alpha, \beta)$ ER linear regenerating codes. The proof for the case of $(5, 4, 4)$ is essentially built on top of the proof for the case of $(4, 3, 3)$ and involves one additional idea which was not present (rather, not needed) in the case of $(4, 3, 3)$. As we will see in Section \ref{sec:general}, this extra step is the key to deriving the proof for any general $n$. We will keep the focus of this section on the main steps of the proof without giving detailed calculations, wherever such calculations resemble those for the case of $(4, 3, 3)$.  Like in the case of $(4, 3, 3)$, we begin with a re-statement of Theorem \ref{thm:new_bound_k_eq_d} for the case of $(5, 4, 4)$, where the result is stated in terms of $\text{rank}(H)$. 

\vspace{0.1in}

\begin{thm} \label{thm:544}
Consider an exact repair linear regenerating code $\mathcal{C}$, having parameters $(n=5, k=4, d=4,), (\alpha, \beta)$. Let the matrix $H$ correspond to the dual of the code $\mathcal{C}$. Then, the rank of the matrix $H$ is lower bounded by
\begin{eqnarray} \label{eq:bound_rank_H_544}
\text{rank}(H) & \geq & \left \{ \begin{array}{c} \left \lceil \frac{10(\alpha - \beta)}{3} \right \rceil, \ 2\beta \leq \alpha \leq 4\beta \\  \\
\left \lceil \frac{15\alpha - 10\beta}{6} \right \rceil, \ \frac{4}{3}\beta \leq \alpha \leq 2\beta \\  \\
2\alpha - \beta, \ \beta \leq \alpha \leq \frac{4}{3}\beta \end{array} \right. .
\end{eqnarray}
\end{thm}

\vspace{0.1in}

Like in the of case of $(4, 3, 3)$, the fact $\text{rank}(H)  \geq 2\alpha - \beta, \beta \leq \alpha \leq 2\beta $ follows from our proof of FR-trade-off, see \eqref{eq:dimakis_proof_4}. Thus, notice that in order to prove Theorem \ref{thm:544}, it suffices to prove the following bounds on $\text{rank}(H)$ individually, without considering any particular range of $\alpha$:
\begin{eqnarray} 
\textit{Bound  1} : \ \ \text{rank}(H)  & \geq & \left \lceil \frac{10(\alpha - \beta)}{3} \right \rceil \label{eq:bound1_544} \\
\textit {Bound 2} :  \ \ \text{rank}(H)  & \geq & \left \lceil \frac{15\alpha - 10\beta}{6} \right \rceil \label{eq:bound2_544}.
\end{eqnarray}

We will now separately illustrate the main steps involved in the proofs of \eqref{eq:bound1_544} and \eqref{eq:bound2_544}. 

\vspace{0.1in}

\subsection{Proof of \eqref{eq:bound1_544}} \label{sec:proof_bound1_544}

The bound in \eqref{eq:bound1_544} can be derived exactly in the same way we derived \eqref{eq:433_proof_4} for the case of $(4, 3, 3)$. Thus, we define the matrices $H^{(5)}$ and $H^{(4)}$ (in the same way we defined the matrices $H^{(4)}$ and $H^{(3)}$ for the case of $(4, 3, 3)$) and make the comparison $\rho\left(H^{(5)}\right) \geq \rho\left(H^{(4)}\right)$. More formally, the matrices are defined as follows :
\begin{eqnarray}
H^{(5)} & = & H_{repair},
\end{eqnarray}
where $H_{repair}$ is as given by \eqref{eq:Hrepair}. Also, let the matrix $H^{(5)}_j$ denote the $j^{\text{th}}$ thick column of $H^{(5)}, 1 \leq j \leq 5$, i.e., $H^{(5)} = [H^{(5)}_1 \ H^{(5)}_2  H^{(5)}_3 \ H^{(5)}_4 \ H^{(5)}_5]$. Next, define the matrices $H^{(4)}_j,  2 \leq j \leq 5$, such that the columns of $H^{(4)}_j$ form a basis for the vector space $\mathcal{S}\left(H^{(5)}_j\right) \cap \mathcal{S}\left(H^{(5)}|_{[j-1]}\right)$. Also define the matrix $H^{(4)}$ as
\begin{eqnarray} \label{eq:proof544_a}
H^{(4)} = [H^{(4)}_2  H^{(4)}_3 \ H^{(4)}_4 \ H^{(4)}_5].
\end{eqnarray}
Both the matrices $H^{(5)}_j$ and $H^{(4)}_j$ are also associated with corresponding block-submatrix representations, i.e., 
\begin{eqnarray}
H^{(5)} & = & \left( A^{(5)}_{i,j} , 1 \leq i, j \leq 5 \ \right ),  \\
H^{(4)} & = & \left( A^{(4)}_{i,j} , 1 \leq i \leq 5, 2 \leq j \leq 5 \ \right ). \label{eq:proof544_b}
\end{eqnarray}
The calculation of $\rho\left(H^{(5)}\right)$ and $\rho\left(H^{(4)}\right)$ follow steps similar to those in \eqref{eq:433_proof_5z}-\eqref{eq:433_proof_7} and \eqref{eq:433_proof_8}-\eqref{eq:proof_433_rankH3_short},respectively.
The subsequent comparison of the two ranks (i.e, $\rho(H^{(5)}) \geq \rho(H^{(4)})$) yields the following lower bound on the rank of $H^{(5)}$:
\begin{eqnarray}
\rho\left( H^{(5)} \right) & \geq & \frac{1}{3}\left\{ 2\sum_{j=1}^{5}\rho\left(A^{(5)}_{j,j}\right) - \sum_{j=2}^{5} \sum_{\ell=1}^{j-1}\rho\left(A^{(5)}_{j,\ell}\right)\right\} \label{eq:proof544_rankH5}.
\end{eqnarray}
Note that \eqref{eq:proof544_rankH5} is the analogue of the bound in Theorem \ref{thm:433_rankH4} (see \ref{eq:433_proof_16}), where  the upper limits of summations have been changed from $4$ to $5$, and $\{A^{(4)}_{i, j}\}$ have been replaced by $\{A^{(5)}_{i,j}\}$. Now, we invoke the facts that $\rho\left(A^{(5)}_{j,j}\right) = \alpha, 1 \leq j \leq 5$ and $\rho\left(A^{(5)}_{i,j}\right) \leq \beta, 1 \leq i, j \leq 5, \ i \neq j$. Thus, we get 
\begin{eqnarray}
\rho\left( H^{(5)} \right) & \geq &  \frac{1}{3}\left\{ 2\sum_{j=1}^{5}\alpha - \sum_{j=2}^{5} (j-1)\beta\right\} \\
& = &  \frac{10(\alpha - \beta)}{3}. 
\end{eqnarray}
This completes the proof of \eqref{eq:bound1_544}.

\vspace{0.1in}

\subsection{Proof of \eqref{eq:bound2_544}} \label{sec:proof_bound2_544}

This is the part which is new in the case of $(5, 4, 4)$. For proving \eqref{eq:bound2_544}, consider the matrix $H^{(4)}$ as defined in \eqref{eq:proof544_a}, and whose block submatrix representation is as given by \eqref{eq:proof544_b}, i.e.,
\begin{eqnarray}
H^{(4)} & = & [H^{(4)}_2 \ H^{(4)}_3 \ H^{(4)}_4 \ H^{(4)}_5  ] \ = \ \left( A^{(4)}_{i,j} , 1 \leq i \leq 5, 2 \leq j \leq 5 \ \right ), \label{eq:proof544_c}
\end{eqnarray}
Define the submatrix $\tilde{H}^{(4)}$ of $H^{(4)}$ as follows:
\begin{eqnarray}
\tilde{H}^{(4)} & = & \left( A^{(4)}_{i,j} , 2 \leq i \leq 5, 2 \leq j \leq 5 \ \right ), 
\end{eqnarray}
i.e., $\tilde{H}^{(4)}$ is formed by excluding the first thick row of $H^{(4)}$. The matrix $\tilde{H}^{(4)}$ has exactly $4$ thick rows and 4 thick columns. Next, observe that the matrix  $H^{(4)}$ satisfies the following properties:
\begin{enumerate}[1.]
\item For all $j, 2 \leq j \leq 5$, the columns of $H^{(4)}_j$ are linearly independent, and
\item $\rho\left(H^{(4)}_j\right) = \rho\left(A^{(4)}_{j,j}\right), \ 2 \leq j \leq 5$. 
\end{enumerate}
The first property follows directly from the definition of $H^{(4)}_j$ (since we know that the columns form a basis for a certain subspace), while the second property is the analogue of \eqref{eq:433_proof_0} for the case of $(5, 4, 4)$. Also, note that the above two properties together imply that the matrix $A^{(4)}_{j,j}, 2 \leq j \leq 5$ has full column rank. It then follows that the sub-matrix $\tilde{H}^{(4)}$ satisfies the following properties :
\begin{enumerate}[1.]
\item $\rho\left( \tilde{H}^{(4)}_j \right) = \rho\left( A^{(4)}_{j,j} \right), \ 2 \leq j \leq 5$. This follows as a result of the fact that the  matrix $A^{(4)}_{j,j}$ has full column rank.
\item For all $j, 2 \leq j \leq 5$, the columns of $\tilde{H}^{(4)}_j$ are all linearly independent. This follows because  1) the previous statement implies that $\rho\left( \tilde{H}^{(4)}_j \right) = \rho\left( H^{(4)}_j \right)$, and 2) the columns of  $H^{(4)}_j$ are linearly independent.
\end{enumerate}
Remark \ref{rem:resuse433} now implies that Theorem \ref{thm:433_rankH4} can be used to lower bound the rank of the matrix $\tilde{H}^{(4)}$. The only thing that needs to be taken care of is the fact that in Theorem $\ref{thm:433_rankH4}$, the indices of the thick columns range from $1$ to $4$, where as for the matrix $\tilde{H}^{(4)}$, they range from $2$ to $5$. Accounting for this variation in Theorem \ref{thm:433_rankH4}, we get
\begin{eqnarray}
\rho\left( H^{(4)} \right) \  \geq \ \rho\left( \tilde{H}^{(4)} \right) & \geq & \frac{1}{3}\left\{ 2\sum_{j=2}^{5}\rho\left(A^{(4)}_{j,j}\right) - \sum_{j=3}^{5} \sum_{\ell=2}^{j-1}\rho\left(A^{(4)}_{j,\ell}\right)\right\}. \label{eq:proof544_d}
\end{eqnarray}
In summary, we have a got a new lower bound on the rank of $H^{(4)}$ and this new bound in \eqref{eq:proof544_d} is, in general, different from what is used to prove \eqref{eq:bound1_544}. We will now show that \eqref{eq:bound2_544} is obtained as necessary condition  for satisfying $\rho(H^{(5)}) \geq \rho(H^{(4)})$, where  $\rho(H^{(4)})$ is assumed to be lower bounded as in \eqref{eq:proof544_d}. Towards this, note from \eqref{eq:proof544_rankH5} that $\rho(H^{(5)})$ is  given by
\begin{eqnarray}
\rho\left( H^{(5)} \right) & = & \rho\left(A^{(5)}_{1,1}\right) \ + \  \sum_{j=2}^{5} \left\{ \left(\rho\left(A_{j,j}^{(5)}\right) - \sum_{\ell=1}^{j-1}\rho\left(A_{j,\ell}^{(5)}\right)\right)^+  +   \alpha_j \right\}. \label{eq:proof544_e} 
\end{eqnarray}
Recall that \eqref{eq:proof544_e} is the same expression for $\rho(H^{(5)})$ that is used in the proof of $\eqref{eq:bound1_544}$. Given the two bounds on $\rho(H^{(5)})$ and $\rho(H^{(4)})$, the remaining sequence of steps that need to be carried out to complete proof of \eqref{eq:bound2_544} are similar to what we do for the proof of \eqref{eq:bound1_544}. These are given as follows: 
\begin{enumerate}[1.]
\item Express the quantities $\{A^{(4)}_{i,j}\}$ appearing in  \eqref{eq:proof544_d} in terms of $\{A^{(5)}_{i,j}\}$. This is accomplished via an analogue of Lemma \ref{lem:433_intersections} for the case of $(5, 4, 4)$.
\item Next, obtain a lower bound on $\sum_{j=2}^{5}\alpha_j$ by invoking the comparison 
$\rho(H^{(5)}) \geq \rho(H^{(4)})$, where $\rho(H^{(5)})$ and $\rho(H^{(4)})$ are given by \eqref{eq:proof544_e} and \eqref{eq:proof544_d}, respectively.  
\item Finally, use the lower bound on $\sum_{j=2}^{5}\alpha_j$ back in \eqref{eq:proof544_e} to get a lower bound on $\rho(H^{(5)})$. 
\end{enumerate}
We will defer the calculations of the above three steps until Section \ref{sec:general}, where we give the full proof of Theorem \ref{thm:new_bound_k_eq_d}. As we will see, the bound on $\rho(H^{(5)})$ that is obtained by after carrying out the above steps is given by
\begin{eqnarray}
\rho\left( H^{(5)} \right) & \geq & \frac{1}{6}\left\{ 3\sum_{j=1}^{5}\rho\left(A^{(5)}_{j,j}\right) - \sum_{j=2}^{5} \sum_{\ell=1}^{j-1}\rho\left(A^{(5)}_{j,\ell}\right)\right\} \label{eq:proof544_f}.
\end{eqnarray}
Given  \eqref{eq:proof544_f}, we invoke the facts that $\rho\left(A^{(5)}_{j,j}\right) = \alpha, 1 \leq j \leq 5$ and $\rho\left(A^{(5)}_{i,j}\right) \leq \beta, 1 \leq i, j \leq 5, \ i \neq j$. Thus, we get 
\begin{eqnarray}
\rho\left( H^{(5)} \right) & \geq &  \frac{1}{6}\left\{ 3\sum_{j=1}^{5}\alpha - \sum_{j=2}^{5} (j-1)\beta\right\} \\
& = &  \frac{15\alpha - 10\beta}{6}. 
\end{eqnarray}
This completes the proof of \eqref{eq:bound2_544}.

\vspace{0.1in}

\begin{note}
We would like to mention that the equivalent of Theorem \ref{thm:433_rankH4} for the case of $(5, 4, 4)$ involves putting together the three equations \eqref{eq:proof544_rankH5}, \eqref{eq:proof544_d} and \eqref{eq:proof544_f}.
\end{note}

\vspace{0.1in}

\section{Proof of Theorem \ref{thm:new_bound_k_eq_d} for general $(n, k = n-1, d = n-1)$} \label{sec:general}

\vspace{0.1in}

In this section, we will prove Theorem \ref{thm:new_bound_k_eq_d} for any general $(n, k = n-1, d = n-1)$. We begin with a re-statement of Theorem \ref{thm:new_bound_k_eq_d} in terms of rank of $H$.

\vspace{0.1in}

\begin{thm} \label{thm:rankH_new_bound_k_eq_d}
Consider an exact repair linear regenerating code $\mathcal{C}$, having parameters $(n, k = n-1, d = n-1), (\alpha, \beta)$. Let the matrix $H$ correspond to the dual of the code $\mathcal{C}$. Then, the rank of the matrix $H$ is lower bounded by
\begin{eqnarray} \label{eq:bound_rank_H_gen}
\text{rank}(H) & \geq & \left \{ \begin{array}{c} \left \lceil \frac{2rn\alpha - n(n-1)\beta}{r^2+r}\right \rceil, \ \frac{d\beta}{r} \leq \alpha \leq \frac{d\beta}{r-1}, \text{ where } 2 \leq r \leq n - 2 \\  \\
2\alpha - \beta, \ \frac{d\beta}{n-1} \leq \alpha \leq \frac{d\beta}{n-2} \end{array} \right. .
\end{eqnarray}
\end{thm}

\vspace{0.1in}

Like we mentioned in the special cases of $(4, 3, 3)$ and $(5, 4, 4)$, the fact $\text{rank}(H)  \geq 2\alpha - \beta, \beta \leq \alpha \leq 2\beta $ follows from our proof of FR-trade-off, see \eqref{eq:dimakis_proof_4}. Thus, in order to prove Theorem \ref{thm:rankH_new_bound_k_eq_d}, it suffices to prove the following bound on $\text{rank}(H)$, without considering any particular range of $\alpha$:
\begin{eqnarray} 
\text{rank}(H)  & \geq & \left \lceil \frac{2rn\alpha - n(n-1)\beta}{r^2+r}\right \rceil, 2 \leq r \leq n - 2. \label{eq:general_bound_to_prove}
\end{eqnarray}
Note that there are in fact $n-3$ bounds which needs to be established for the general case, where each bound corresponds to a value of $r$ in the range $ 2 \leq r \leq n-2$. A quick outline of the proof is provided next. We will consider the matrices $H^{(t)}, 3 \leq t \leq n$, where $H^{(n)} = H_{repair}$ corresponding to the matrix $H$, and where the matrix $H^{(t)}, 3 \leq t \leq n-1$ is defined based on matrix $H^{(t+1)}$. We will also have the relation $\rho(H^{(t)})   \geq  \rho(H^{(t-1)}), 3 \leq t \leq n$. The bound in \eqref{eq:general_bound_to_prove} corresponding to the case of a general $r, 2 \leq r \leq n - 2$ will be obtained as necessary condition for satisfying the chain of inequalities given by 
\begin{eqnarray}
\rho(H^{(n)})  & \geq & \rho(H^{(n-1)}) \ \geq \ \rho(H^{(n-2)}) \ \geq \ \ldots \ \geq \rho(H^{(n-(r-2))}) \ \geq \ \rho(H^{(n-(r-1))}).
\end{eqnarray}
Recall that for the case of $(5, 4, 4)$, we used the relations 1) $\rho(H^{(5)}) \geq  \rho(H^{(4)})$ and 2) $\rho(H^{(5)})  \geq  \rho(H^{(4)}) \geq  \rho(H^{(3)})$ to prove \eqref{eq:bound1_544} and \eqref{eq:bound2_544}, respectively.

The rest of the technical discussion in this section is divided as follows:
\begin{enumerate}[1.]
\item Formal definition of the matrices $H^{(t)}, 3 \leq t \leq n$, and also the associated block sub-matrix representations.
\item Establish a generalization of Lemma \ref{lem:433_intersections} which in turn enables us to prove a generalization of Theorem \ref{thm:433_rankH4}. Arguments based on mathematical induction will be used in the proofs of both these generalizations. 
\item Deriving the bound in \eqref{eq:general_bound_to_prove} based on the generalization of Theorem \ref{thm:433_rankH4}.
\end{enumerate}

\vspace{0.1in}

\subsection{Notation}

\subsubsection{The Matrices $\{H^{(t)}\}$}

As stated above, we assume that $H^{(n)} = H_{repair}$, where $H_{repair}$ is as defined by Lemma \ref{lem:H_repair}. Also assume that $H^{(n)}_j$ denotes the $j^{\text{th}}$ thick column of $H^{(n)}, 1 \leq j \leq n$, i.e., 
\begin{eqnarray} \label{eq:proofgen_Hndef}
H^{(n)} & = & [H^{(n)}_1 \ H^{(n)}_2  \ \ldots \  H^{(n)}_n].
\end{eqnarray}
The matrices $H^{(t)}, 3 \leq t \leq n-1$ are iteratively defined as follows: 

\vspace{0.1in}

\begin{framed}
\begin{enumerate}[Step 1.]
\item Let $ t=n-1$.
\item  Define the matrices $H^{(t)}_j,  n-t+1 \leq j \leq n$, such that the columns of $H^{(t)}_j$ form a basis for the vector space $\mathcal{S}\left(H^{(t+1)}_j\right) \cap \mathcal{S}\left(H^{(t+1)}|_{\{n-t, n-t+1, \ldots, j-1\}}\right)$.
\item  Define the matrix $H^{(t)}$ as
\begin{eqnarray} \label{eq:proofgen_Htdef}
H^{(t)} = [H^{(t)}_{n-t+1}  \ H^{(t)}_{n-t+2} \ \ldots \ H^{(t)}_n].
\end{eqnarray}
\item  If $t \geq 4$, decrement $t$ by $1$ and go back to Step $2$. 
\end{enumerate}
\end{framed}

Steps $2$, $3$, $4$ are carried out in that sequence a total of $n-3$ times so that all the matrices $H^{(t)}, 3 \leq t \leq n-1$ get defined. Clearly, the ranks of the matrices $H^{(t)}, 3 \leq t \leq n$ are ordered as
\begin{eqnarray} \label{eq:proofgen_rankorderHs}
\rho(H^{(n)})  & \geq & \rho(H^{(n-1)}) \ \geq \ \ \ldots \ \geq \rho(H^{(4)}) \ \geq \ \rho(H^{(3)}).
\end{eqnarray}

\vspace{0.1in}

\subsubsection{Block Sub-Matrix Representation of $H^{(t)}$}

The block submatrix representation for the matrix $H^{(n)}$ is given by \eqref{eq:Hrepair}. Like we did in the case of $(4, 3, 3)$ and $(5, 4, 4)$, for easiness of notation, \eqref{eq:Hrepair} will be denoted as 
\begin{eqnarray}
H^{(n)} & = & \left( A^{(n)}_{i,j} , 1 \leq i, j \leq n \ \right ),
\end{eqnarray}
where $A^{(n)}_{i,i} = I_{\alpha}, 1 \leq i \leq n$, and when $i \neq j$, we have added a superscript on $A_{i,j}$. 
Next, we introduce block sub-matrix representations for the other $n-3$ matrices $H^{(t)}, 3 \leq t \leq n-1$. The matrix $H^{(t)}$ will be identified as 
\begin{eqnarray}
H^{(t)} & = & \left( A^{(t)}_{i,j} , 1 \leq i \leq n, n-t+1 \leq j \leq n \ \right ),
\end{eqnarray}
where  $A^{(t)}_{i,j}$ is an $\alpha \times \rho(H^{(t)}_j)$ matrix (over $\mathbb{F}_q$) such that 
\begin{eqnarray}
\mathcal{S}\left(A^{(t)}_{i,j} \right) & \subseteq & \mathcal{S}\left(A^{(t+1)}_{i,j} \right) \ \bigcap \ \sum_{\ell=n-t}^{j-1}\mathcal{S}\left(A^{(t+1)}_{i,\ell}\right).  \label{eq:proofgen_1}
\end{eqnarray}
Note that \eqref{eq:proofgen_1} is a direct consequence of the definition of the matrix $H^{(t)}$. The following lemma gives additional properties of the matrices $\{A^{(t)}_{i,j}\}$. While the first part of the lemma is a generalization of \eqref{eq:433_proof_0}, the second and third parts together generalize Lemma \ref{lem:433_intersections}.

\vspace{0.1in}

\begin{lem} \label{lem:nkk_intersections}
\begin{enumerate}[a)]
\item
\begin{eqnarray} 
\rho\left(H^{(t)}_j\right) & = & \rho\left( A^{(t)}_{j,j}\right),  \ 3 \leq t \leq n, \ n-t+1 \leq j \leq n. \label{eq:proofgen_2}
\end{eqnarray}
\item 
\begin{eqnarray}
\rho\left(A^{(t)}_{j,j} \right) & = & \rho\left(A^{(t+1)}_{j,j} \right) \ - \ \left\{ \rho\left(H^{(t+1)}|_{\{n-t, \ldots, j\}}\right) - \rho\left(H^{(t+1)}|_{\{n-t, \ldots, j-1\}}\right) \right\}, \nonumber \\
& & \hspace{2in} \ 3 \leq t \leq n-1, \ n-t+1 \leq j \leq n. \label{eq:proofgen_3}
\end{eqnarray}
\item   
\begin{eqnarray}
\sum_{\ell=n-t+1}^{j-1} \rho\left(A^{(t)}_{j,\ell} \right) & \leq & \sum_{\ell=n-t}^{j-1}\rho\left(A^{(t+1)}_{j,\ell} \right) \ - \ \rho\left(A^{(t)}_{j,j} \right), \ 3 \leq t \leq n-1, \ n-t+2 \leq j \leq n. \label{eq:proofgen_4}
\end{eqnarray}
\end{enumerate}
\end{lem}
\begin{proof}
See Appendix \ref{app:Prooflem}	
\end{proof}	

\vspace{0.1in}

\vspace{0.1in}
\subsection{Generalization of  Theorem \ref{thm:433_rankH4}}

\begin{thm} \label{thm:nkk_rankviainduction}
Consider the matrices $\{H^{(t)}, 4 \leq t \leq n \}$ as defined together by \eqref{eq:proofgen_Hndef} and \eqref{eq:proofgen_Htdef}. Also, consider the associated block sub-matrix representations given by  $H^{(t)}  = \left( A^{(t)}_{i,j} , 1 \leq i \leq n, n-t+1 \leq j \leq n \ \right ),  4 \leq t \leq n$. Then, for any $s$ such that $1 \leq s \leq n-3$, and any $t$ such that $3+s \leq t \leq n$, the rank of the matrix $H^{(t)}$ is lower bounded by
\begin{eqnarray}
\rho\left( H^{(t)} \right) & \geq & \frac{2}{(s+1)(s+2)}\left\{ (s+1)\sum_{j=n-t+1}^{n}\rho\left(A^{(t)}_{j,j}\right) - \sum_{j=n-t+2}^{n} \sum_{\ell=n-t+1}^{j-1}\rho\left(A^{(t)}_{j,\ell}\right)\right\}. \label{eq:boundgen_rankH}
\end{eqnarray}
\end{thm}
\begin{proof}
A proof based on induction on the parameter $s$ appears in Appendix \ref{app:ProofThm}. The induction starts at $s=1$ and ends at $s = n-3$. 
\end{proof}	

\vspace{0.1in}

Observe that if we set $n=4$ in Theorem \ref{thm:nkk_rankviainduction}, there is only one $(s, t)$ pair for which we get a bound via the theorem. The pair is given by $(s = 1, t = 4)$, and in this case, the bound in \eqref{eq:boundgen_rankH} is exactly  same as what we have in \eqref{eq:433_proof_16r}. For the case when $n=5$, we get bounds for three pairs of $(s, t)$ given by $(s=1, t=5)$, $(s=1, t = 4)$ and $(s = 2, t = 5)$. The bounds obtained from Theorem \ref{thm:nkk_rankviainduction} for these three cases are given by \eqref{eq:proof544_rankH5}, \eqref{eq:proof544_d} and \eqref{eq:proof544_f}, respectively.

\subsection{Proof of Theorem \ref{thm:rankH_new_bound_k_eq_d}}

We will now give a proof of \eqref{eq:general_bound_to_prove} based on Theorem \ref{thm:nkk_rankviainduction}. Recall from our earlier discussion that proving \eqref{eq:general_bound_to_prove} is sufficient to prove Theorem \ref{thm:rankH_new_bound_k_eq_d}. 
For proving \eqref{eq:general_bound_to_prove}, we evaluate the bound in \eqref{eq:boundgen_rankH} for the $(n-3)$ pairs given by $(s, t=n), 1 \leq s \leq n-3$. Further, we also invoke the facts that $\rho\left(A^{(n)}_{j,j}\right) = \alpha, 1 \leq j \leq n$ and $\rho\left(A^{(n)}_{i,j}\right) \leq \beta, 1 \leq i, j \leq n, i \neq j$. Thus, we get that 
\begin{eqnarray}
\rho\left( H^{(n)} \right) & \geq & \frac{2}{(s+1)(s+2)}\left\{ (s+1)\sum_{j=1}^{n}\rho\left(A^{(n)}_{j,j}\right) - \sum_{j=2}^{n} \sum_{\ell=1}^{j-1}\rho\left(A^{(n)}_{j,\ell}\right)\right\} \\
& \geq & \frac{2}{(s+1)(s+2)}\left\{ (s+1)\sum_{j=1}^{n}\alpha - \sum_{j=2}^{n} (j-1)\beta\right\} \\
& = & \frac{2(s+1)n\alpha - n(n-1)\beta}{(s+1)(s+2)}, \ 1 \leq s \leq n-3. \label{eq:proofgenz}
\end{eqnarray}
Finally, note that \eqref{eq:general_bound_to_prove} follows from \eqref{eq:proofgenz} by substituting $r = s+1$. This completes the proof of \eqref{eq:general_bound_to_prove} and thereby, also the proof of Theorem \ref{thm:rankH_new_bound_k_eq_d}. 

\bibliographystyle{IEEEtran}
\bibliography{regen_codes_tradeoff}

\appendices

\section{Proof of Lemma \ref{lem:nkk_intersections}} \label{app:Prooflem}

The proofs of \eqref{eq:proofgen_2}, \eqref{eq:proofgen_3} and \eqref{eq:proofgen_4} are very similar to the proofs of 
\eqref{eq:433_proof_0}, \eqref{eq:433_proof_8} and \eqref{eq:433_proof_10}, respectively. We will still give a full proof of Lemma \ref{lem:nkk_intersections} here for sake of completeness.

\subsection{Proof of a) : }
We will give a proof of \eqref{eq:proofgen_2} based on an induction argument on the parameter $t$, starting at $t=n$ and decrementing $t$ by $1$ at each step. As for the induction start, the fact that $\rho(H^{(n)}_j) = \rho(A^{(n)}_{j,j})$ follows because we had defined  $H^{(n)} = H_{repair}$, and from \eqref{eq:Hrepair} it is clear that
\begin{eqnarray}
\rho\left(H^{(n)}_j\right) & = & \rho\left( A^{(n)}_{j,j}\right)  \ = \ \alpha, \ 1 \leq j \leq n.
\end{eqnarray}
Next, assume that
\begin{eqnarray}
\rho\left(H^{(t+1)}_j\right) & = & \rho\left( A^{(t+1)}_{j,j}\right), \ n-t \leq j \leq n \label{eq:prooflemm_app_1}
\end{eqnarray}
for some $t$ in the range $3 \leq t \leq n-1$. We will now prove that
\begin{eqnarray}
\rho\left(H^{(t)}_j\right) & = & \rho\left( A^{(t)}_{j,j}\right), \ n-t+1 \leq j \leq n. \label{eq:prooflemm_app_2}
\end{eqnarray}
Towards this, recall from the definition of $H^{(t)}_j$ that the columns of $H^{(t)}_j$ form a basis for the vector space
$\mathcal{S}\left(H^{(t+1)}_j\right) \cap \mathcal{S}\left(H^{(t+1)}|_{\{n-t, n-t+1, \ldots, j-1\}}\right)$.  Now, first of all note that the $\rho(H^{(t)}_j)$ columns of $H^{(t)}_j$ can be extended to a basis for $\mathcal{S}\left(H^{(t+1)}_j\right)$ by adding exactly $\rho\left(H^{(t+1)}_j\right) - \rho\left(H^{(t)}_j\right)$ additional columns. If we now focus on the sub-matrix $A^{(t)}_{j,j}$, it then follows that the $\rho(H^{(t)}_j)$ columns of $A^{(t)}_{j,j}$ can be extended to a basis for $\mathcal{S}\left(A^{(t+1)}_{j,j}\right)$ by adding at most $\rho\left(H^{(t+1)}_j\right) - \rho\left(H^{(t)}_j\right)$ additional columns. From the induction hypothesis in \eqref{eq:prooflemm_app_1}, we know that $\rho\left(H^{(t+1)}_j\right)  =  \rho\left( A^{(t+1)}_{j,j}\right)$, and hence it must also be true that $\rho\left(H^{(t)}_j\right) =  \rho\left( A^{(t)}_{j,j}\right), \ n-t+1 \leq j \leq n$.

\vspace{0.1in}

\subsection{Proof of b) : }

For the proof of \eqref{eq:proofgen_3} appearing in Lemma \ref{lem:nkk_intersections}, observe the following chain of equalities:
\begin{eqnarray}
\rho\left(H^{(t+1)}|_{\{n-t, \ldots, j\}}\right) - \rho\left(H^{(t+1)}|_{\{n-t, \ldots, j-1\}}\right) &  & \\
& \hspace{-4in} = & \hspace{-2in} \text{dim}\left(\mathcal{S}\left(H^{(t+1)}|_{\{n-t, \ldots, j-1\}} \right) + \mathcal{S}\left(H^{(t+1)}_j \right)\right) -  \text{dim}\left(\mathcal{S}\left(H^{(t+1)}|_{\{n-t, \ldots, j-1\}}\right)\right) \label{eq:nkk_proof_11} \\ 
& \hspace{-4in} = & \hspace{-2in} \text{dim}\left(\mathcal{S}\left(H^{(t+1)}_j \right)\right) - \text{dim}\left(\mathcal{S}\left(H^{(t+1)}|_{\{n-t, \ldots, j-1\}} \right) \cap \mathcal{S}\left(H^{(t+1)}_j \right)\right) \label{eq:nkk_proof_12} \\
& \hspace{-4in} = & \hspace{-2in} \rho\left(H^{(t+1)}_j \right) - \rho\left(H^{(t)}_j \right) \label{eq:nkk_proof_13} \\
& \hspace{-4in} = & \hspace{-2in} \rho\left(A^{(t+1)}_{j,j} \right) - \rho\left(A^{(t)}_{j,j} \right), \label{eq:nkk_proof_14}
\end{eqnarray} 
where \eqref{eq:nkk_proof_12} follows from the fact that for any two subspaces $W_1, W_2$, $\text{dim}(W_1 + W_2)  = \text{dim}(W_1) + \text{dim}(W_2) - \text{dim}(W_1 \cap W_2)$. Equation  \eqref{eq:nkk_proof_13} follows from the definition of $H^{(t)}_j$, while \eqref{eq:nkk_proof_14} follows from Part a) of Lemma \ref{lem:nkk_intersections}. This completes the proof of \eqref{eq:proofgen_3}. 

\vspace{0.1in}

\subsection{Proof of c) : }
We need to prove that 
\begin{eqnarray}
\sum_{\ell=n-t+1}^{j-1} \rho\left(A^{(t)}_{j,\ell} \right) & \leq & \sum_{\ell=n-t}^{j-1}\rho\left(A^{(t+1)}_{j,\ell} \right) \ - \ \rho\left(A^{(t)}_{j,j} \right), \ 3 \leq t \leq n-1, \ n-t+2 \leq j \leq n. \label{eq:proofgen_4a}
\end{eqnarray}

Towards this, observe from \eqref{eq:proofgen_1} that we have $\mathcal{S}\left(A^{(t)}_{j,j}\right) \subseteq  \sum_{\ell = n-t}^{j-1}\mathcal{S}\left(A^{(t+1)}_{j,\ell}\right)$ and thus, we get that
\begin{eqnarray}
\rho \left(A^{(t)}_{j,j}\right) & \leq &  \text{dim} \left( \sum_{\ell = n-t}^{j-1}\mathcal{S}\left(A^{(t+1)}_{j,\ell}\right) \right), \ 3 \leq t \leq n-1, \ n-t+2 \leq j \leq n. \label{eq:433_proof_15aa}
\end{eqnarray}
The right hand side of \eqref{eq:433_proof_15aa} can be upper bounded as follows:
\begin{eqnarray}
\text{dim} \left( \sum_{\ell = n-t}^{j-1}\mathcal{S}\left(A^{(t+1)}_{j,\ell}\right) \right) & = & \text{dim} \left( \sum_{\ell = n-t}^{j-2}\mathcal{S}\left(A^{(t+1)}_{j,\ell}\right)  + \mathcal{S}\left(A^{(t+1)}_{j,j-1}\right)\right) \label{eq:433_proof_16aa} \\
& \hspace{-2in} = & \hspace{-1in} \text{dim} \left( \sum_{\ell = n-t}^{j-2}\mathcal{S}\left(A^{(t+1)}_{j,\ell}\right) \right) + \text{dim} \left(  \mathcal{S}\left(A^{(t+1)}_{j,j-1}\right)\right) - \text{dim} \left( \sum_{\ell = n-t}^{j-2}\mathcal{S}\left(A^{(t+1)}_{j,\ell}\right)  \cap \mathcal{S}\left(A^{(t+1}_{j,j-1}\right)\right) \label{eq:433_proof_17aa} \\
& \hspace{-2in} \leq & \hspace{-1in} \text{dim} \left( \sum_{\ell = n-t}^{j-2}\mathcal{S}\left(A^{(t+1)}_{j,\ell}\right) \right) + \text{dim} \left(  \mathcal{S}\left(A^{(t+1)}_{j,j-1}\right)\right) - \text{dim} \left( \mathcal{S}\left(A^{(t)}_{j,j-1}\right)\right), \label{eq:433_proof_18aa} \\
& \hspace{-2in} = & \hspace{-1in} \text{dim} \left( \sum_{\ell = n-t}^{j-2}\mathcal{S}\left(A^{(t+1)}_{j,\ell}\right) \right) + \rho\left(A^{(t+1)}_{j,j-1}\right) - \rho\left(A^{(t)}_{j,j-1}\right), \label{eq:433_proof_19aa}
\end{eqnarray}
where \eqref{eq:433_proof_18aa} follows from \eqref{eq:proofgen_1}. The term $\text{dim} \left( \sum_{\ell = n-t}^{j-2}\mathcal{S}\left(A^{(t+1)}_{j,\ell}\right) \right)$ appearing in \eqref{eq:433_proof_19aa} can be further upper bounded  by following a similar sequence of steps as in \eqref{eq:433_proof_16aa} - \eqref{eq:433_proof_19aa}. Combining with \eqref{eq:433_proof_15aa}, we eventually get that 
\begin{eqnarray}
\rho \left(A^{(t)}_{j,j}\right) & \leq &   \sum_{\ell=n-t}^{j-1}\rho\left(A^{(t+1)}_{j,\ell} \right) \ - \ \sum_{\ell=n-t+1}^{j-1} \rho\left(A^{(t)}_{j,\ell} \right), \ 3 \leq t \leq n-1, \ n-t+2 \leq j \leq n.  
\end{eqnarray}
This completes the proof of \eqref{eq:proofgen_4a}.

\section{Proof of Theorem \ref{thm:nkk_rankviainduction}} \label{app:ProofThm}

We need to prove that for any $s$ such that $1 \leq s \leq n-3$ and any $t$ such that $3+s \leq t \leq n$, the rank of the   matrix $H^{(t)}$ is lower bounded by
\begin{eqnarray}
\rho\left( H^{(t)} \right) & \geq & \frac{2}{(s+1)(s+2)}\left\{ (s+1)\sum_{j=n-t+1}^{n}\rho\left(A^{(t)}_{j,j}\right) - \sum_{j=n-t+2}^{n} \sum_{\ell=n-t+1}^{j-1}\rho\left(A^{(t)}_{j,\ell}\right)\right\}. \label{eq:boundgen_rankH_repeat}
\end{eqnarray}
The proof will be based on an induction argument on the parameter $s$, starting at $s = 1$ and incrementing $s$ by $1$ at each step.

\vspace{0.1in}

\subsection{Induction Start : The case $s=1$}

We need to prove that for any $t$ such that $4 \leq t \leq n$, the rank of the matrix $H^{(t)}$ is lower bounded by
\begin{eqnarray}
\rho\left( H^{(t)} \right) & \geq & \frac{1}{3}\left\{ 2\sum_{j=n-t+1}^{n}\rho\left(A^{(t)}_{j,j}\right) - \sum_{j=n-t+2}^{n} \sum_{\ell=n-t+1}^{j-1}\rho\left(A^{(t)}_{j,\ell}\right)\right\}. \label{eq:boundgen_rankH_indstart}
\end{eqnarray}
The proof of \eqref{eq:boundgen_rankH_indstart} is very similar to the proof of Theorem \ref{thm:433_rankH4}, where derived the lower bound on $\rho(H^{(4)})$ for the case of $(4, 3, 3)$. Thus, we will first separately calculate (or bound) the column ranks of the matrices $H^{(t)}$ and $H^{(t-1)}$ (for a general $t, 4 \leq t \leq n$) and then show that  \eqref{eq:boundgen_rankH_indstart} is a necessary condition for the satisfying the relation
\begin{eqnarray}
\rho(H^{(t)}) & \geq & \rho(H^{(t-1)}), \ 4 \leq t \leq n. \label{eq:appthmproof_1}
\end{eqnarray}
Recall from \eqref{eq:proofgen_rankorderHs} that the ranks of the matrices $\{\rho(H^{(t)})\}$ are ordered as given in \eqref{eq:appthmproof_1}.

\vspace{0.1in}
\subsubsection{Rank of $H^{(t)}, 4 \leq t \leq n$} Let us define the quantities $\delta_j, n-t+1 \leq j \leq n$ as follows: 
\begin{eqnarray}
\delta_{n-t+1} & = & \rho\left(H^{(t)}_{n-t+1}\right) \ \stackrel{(a)}{=} \ \rho\left(A^{(t)}_{n-t+1,n-t+1}\right), \label{eq:nkk_proof_rankHt_1}\\
\delta_j & = & \rho\left(H^{(t)}|_{\{n-t+1, \ldots, j \}}\right) - \rho\left(H^{(t)}|_{\{n-t+1, \ldots, j-1 \}}\right), \ n-t+2 \leq j \leq n, \label{eq:nkk_proof_rankHt_2}
\end{eqnarray}
where the equality in $(a)$ follows from Part a) of Lemma \ref{lem:nkk_intersections}. It is straightforward to see that  $\delta_j, n-t+2 \leq j \leq n$ is lower bounded by
\begin{eqnarray}
\delta_j & \geq & \left(\rho\left(A_{j,j}^{(t)}\right) - \sum_{\ell=n-t+1}^{j-1}\rho\left(A_{j,\ell}^{(t)}\right)\right)^+, \ n-t+2 \leq j \leq n. \label{eq:nkk_proof_rankHt_3}
\end{eqnarray}
Thus, let us assume that 
\begin{eqnarray}
\delta_j & = & \left(\rho\left(A_{j,j}^{(t)}\right) - \sum_{\ell=n-t+1}^{j-1}\rho\left(A_{j,\ell}^{(t)}\right)\right)^+  +   \alpha_j, \ n-t+2 \leq j \leq n \label{eq:nkk_proof_rankHt_4},
\end{eqnarray}
where $\{\alpha_j, n-t+2 \leq j \leq n\}$ are non-negative integers. The column rank of the matrix $H^{(t)}$ can now be written as
\begin{eqnarray}
\rho\left( H^{(t)} \right) & = & \sum_{j=n-t+1}^{n}\delta_j \\
& = & \rho\left(A^{(t)}_{n-t+1,n-t+1}\right) \ + \  \sum_{j=n-t+2}^{n} \left\{ \left(\rho\left(A_{j,j}^{(t)}\right) - \sum_{\ell=n-t+1}^{j-1}\rho\left(A_{j,\ell}^{(t)}\right)\right)^+  +   \alpha_j \right\} \label{eq:nkk_proof_rankHt_5}.
\end{eqnarray}

\vspace{0.1in}
\subsubsection{Rank of $H^{(t-1)}, 4 \leq t \leq n$} The goal here is to get a lower bound on the rank of $H^{(t-1)}$, where the lower bound depends only on the quantities $\{A^{(t)}_{i, j}\}$ (and not $\{A^{(t-1)}_{i, j}\}$).  The initial steps of this rank calculation here resemble those for the rank calculation of $H^{(t)}$. Thus, let us define the quantities $\gamma_j, n-t+2 \leq j \leq n$ such that
\begin{eqnarray}
\gamma_{n-t+2} & = & \rho\left(H^{(t-1)}_{n-t+2}\right), \\
\gamma_j & = & \rho\left(H^{(t-1)}|_{\{n-t+2, \ldots, j \}}\right) - \rho\left(H^{(t-1)}|_{ \{n-t+2, \ldots, j-1 \} }\right), \ n-t+3 \leq j \leq n. \label{eq:nkk_proof_rankHtp1_1}
\end{eqnarray}
The quantity $\gamma_{n-t+2}$ can be written as 
\begin{eqnarray}
\gamma_{n-t+2} & = & \rho\left(H^{(t-1)}_{n-t+2}\right) \\
& = & \rho\left(A^{(t-1)}_{n-t+2, \ n-t+2}\right) \label{eq:nkk_proof_rankHtp1_2} \\
& = & \rho\left(A^{(t)}_{n-t+2,n-t+2} \right) \ - \ \left\{ \rho\left(H^{(t)}|_{\{n-t+1,\  n-t+2\}}\right) - \rho\left(H^{(t)}_{n-t+1}\right) \right\} \label{eq:nkk_proof_rankHtp1_3} \\
& = & \rho\left(A^{(t)}_{n-t+2,n-t+2} \right) \ - \ \delta_{n-t+2} \label{eq:nkk_proof_rankHtp1_3a} \\
& = &  \rho\left(A^{(t)}_{n-t+2,n-t+2} \right)   \ - \ \left(\rho\left(A_{n-t+2,n-t+2}^{(t)}\right) - \rho\left(A_{n-t+2,n-t+1}^{(t)}\right)\right)^+  -   \alpha_{n-t+2}, \label{eq:nkk_proof_rankHtp1_3az}
\end{eqnarray}
where \eqref{eq:nkk_proof_rankHtp1_2} and \eqref{eq:nkk_proof_rankHtp1_3} respectively follow from Parts a) and b) of Lemma \ref{lem:nkk_intersections}. Also, notice that \eqref{eq:nkk_proof_rankHtp1_3a} and \eqref{eq:nkk_proof_rankHtp1_3az} follow from  \eqref{eq:nkk_proof_rankHt_2} and \eqref{eq:nkk_proof_rankHt_4} respectively. Next, we note that the quantities $\gamma_j, n-t+3 \leq j \leq n $ can be lower bounded in the same way we lower bounded $\delta_j, n-t+2 \leq j \leq n$ in \eqref{eq:nkk_proof_rankHt_3}. Thus, we get that
\begin{eqnarray}
\gamma_j & \geq & \left(\rho\left(A_{j,j}^{(t-1)}\right) - \sum_{\ell=n-t+2}^{j-1}\rho\left(A_{j,\ell}^{(t-1)}\right)\right)^+, \ \ n-t+3 \leq j \leq n. \label{eq:nkk_proof_rankHtp1_3b}
\end{eqnarray}
From Part c) of Lemma \ref{lem:nkk_intersections}, we know that (replace $t$ by $t-1$ in part c) of Lemma \ref{lem:nkk_intersections})
\begin{eqnarray}
\sum_{\ell=n-t+2}^{j-1} \rho\left(A^{(t-1)}_{j,\ell} \right) & \leq & \sum_{\ell=n-t+1}^{j-1}\rho\left(A^{(t)}_{j,\ell} \right) \ - \ \rho\left(A^{(t-1)}_{j,j} \right),  \ n-t+3 \leq j \leq n. \label{eq:nkk_proof_rankHtp1_3c}
\end{eqnarray}
Using \eqref{eq:nkk_proof_rankHtp1_3c} in \eqref{eq:nkk_proof_rankHtp1_3b}, we get that
\begin{eqnarray}
\gamma_j & \geq & \left(2\rho\left(A_{j,j}^{(t-1)}\right) - \sum_{\ell=n-t+1}^{j-1}\rho\left(A_{j,\ell}^{(t)}\right)\right)^+ \label{eq:nkk_proof_rankHtp1_4}\\
& \geq & 2\rho\left(A_{j,j}^{(t-1)}\right) - \sum_{\ell=n-t+1}^{j-1}\rho\left(A_{j,\ell}^{(t)}\right) \label{eq:nkk_proof_rankHtp1_4noplus}\\
& = & 2\left[\rho\left(A^{(t)}_{j,j} \right) \ - \ \left\{ \rho\left(H^{(t)}|_{\{n-t+1, \ldots, j\}}\right) - \rho\left(H^{(t)}|_{\{n-t+1, \ldots, j-1\}}\right) \right\} \right] - \sum_{\ell=n-t+1}^{j-1}\rho\left(A_{j,\ell}^{(t)}\right) \label{eq:nkk_proof_rankHtp1_5} \\
& = & 2\left[\rho\left(A^{(t)}_{j,j} \right) \ - \ \delta_j \right] - \sum_{\ell=n-t+1}^{j-1}\rho\left(A_{j,\ell}^{(t)}\right) \label{eq:nkk_proof_rankHtp1_6} \\
& = & 2\left[\rho\left(A^{(t)}_{j,j} \right) \ - \ \left\{ \left(\rho\left(A_{j,j}^{(t)}\right) - \sum_{\ell=n-t+1}^{j-1}\rho\left(A_{j,\ell}^{(t)}\right)\right)^+  +   \alpha_j \right\} \right] - \sum_{\ell=n-t+1}^{j-1}\rho\left(A_{j,\ell}^{(t)}\right), \ n-t+3 \leq j \leq n, \label{eq:nkk_proof_rankHtp1_7} 
\end{eqnarray}
where \eqref{eq:nkk_proof_rankHtp1_5} follows from part b) of Lemma \ref{lem:nkk_intersections}, while \eqref{eq:nkk_proof_rankHtp1_6} and \eqref{eq:nkk_proof_rankHtp1_7} follow from \eqref{eq:nkk_proof_rankHt_2} and \eqref{eq:nkk_proof_rankHt_4} respectively. The rank of the matrix $H^{(t-1)}$ is now given by $\rho\left(H^{(t-1)}\right) = \sum_{j=n-t+2}^{n}\gamma_j$, where $\gamma_{n-t+2}$ is given by \eqref{eq:nkk_proof_rankHtp1_3az}, and $\{\gamma_j, n-t+3 \leq j \leq n\}$ are lower bounded as given by \eqref{eq:nkk_proof_rankHtp1_7}. We thus get a lower bound on the rank of $H^{(t-1)}$, where the lower bound depends only on the quantities $\{A^{(t)}_{i, j}\}$ and $\{\alpha_j\}$.

\vspace{0.1in}

\subsubsection{Comparison of ranks of $H^{(t)}$ and $H^{(t-1)}$} We are now in a position to show that the bound in  \eqref{eq:boundgen_rankH_indstart} is a necessary condition for satisfying the relation $\rho(H^{(t)}) \geq \rho(H^{(t-1)}), \ 4 \leq t \leq n$. Recall from \eqref{eq:nkk_proof_rankHt_5} that the rank of the matrix $H^{(t)}$ is given by
\begin{eqnarray}
\rho\left( H^{(t)} \right) & = & \rho\left(A^{(t)}_{n-t+1,n-t+1}\right) \ + \  \sum_{j=n-t+2}^{n} \left\{ \left(\rho\left(A_{j,j}^{(t)}\right) - \sum_{\ell=n-t+1}^{j-1}\rho\left(A_{j,\ell}^{(t)}\right)\right)^+  +   \alpha_j \right\}, \ 4 \leq t \leq n. \label{eq:nkk_proof_rankHt_5rep}
\end{eqnarray}
The goal is to obtain a lower bound on $\sum_{j=n-t+2}^{n}\alpha_j$ by invoking the comparison $\rho\left( H^{(t)} \right) \geq \rho\left( H^{(t-1)} \right)$, and then use this lower bound back in \eqref{eq:nkk_proof_rankHt_5rep} to get the desired lower bound on $\rho\left( H^{(t)} \right)$. One can verify that the comparison yields the following lower bound on $\sum_{j=n-t+2}^{n}\alpha_j$:
\begin{eqnarray}
\sum_{j=n-t+2}^{n}\alpha_j & \geq & \frac{1}{3}\left\{ -\rho\left(A^{(t)}_{n-t+1,n-t+1}\right) + \rho\left(A^{(t)}_{n-t+2,n-t+2}\right)  +
2\sum_{j=n-t+3}^{n}\rho\left(A^{(t)}_{j,j}\right) - \right. \nonumber \\ 
& & \hspace{-1in} \left. \left[2\left( \rho\left(A^{(t)}_{n-t+2,n-t+2}\right) - \rho\left(A^{(t)}_{n-t+2,n-t+1}\right)\right)^+     + 3\sum_{j=n-t+3}^{n}\left( \rho\left(A^{(t)}_{j,j}\right) - \sum_{\ell=n-t+1}^{j-1}\rho\left(A^{(t)}_{j,\ell}\right)\right)^+ + \sum_{j=n-t+3}^{n} \sum_{\ell=n-t+1}^{j-1}\rho\left(A^{(t)}_{j,\ell}\right)\right]
\right\} \nonumber \\ \label{eq:nkk_alphasum_indstart}
\end{eqnarray}
It can be verified that substituting  \eqref{eq:nkk_alphasum_indstart} back in \eqref{eq:nkk_proof_rankHt_5rep} results in the following lower bound on rank of $H^{(t)}$:
\begin{eqnarray}
\rho\left( H^{(t)} \right) & \geq & \frac{1}{3}\left\{ 2\sum_{j=n-t+1}^{n}\rho\left(A^{(t)}_{j,j}\right) - \sum_{j=n-t+2}^{n} \sum_{\ell=n-t+1}^{j-1}\rho\left(A^{(t)}_{j,\ell}\right)\right\}.
\end{eqnarray}
This completes our proof of the first step of the induction.

\vspace{0.1in}
\subsection{Induction Step : From $s$ to $s+1$}

Let us assume that for some $s$ in the range $1 \leq s \leq n-4$ and for any $t$ such that $3+s \leq t \leq n$, the rank of the matrix $H^{(t)}$ is lower bounded by
\begin{eqnarray}
\rho\left( H^{(t)} \right) & \geq & \frac{2}{(s+1)(s+2)}\left\{ (s+1)\sum_{j=n-t+1}^{n}\rho\left(A^{(t)}_{j,j}\right) - \sum_{j=n-t+2}^{n} \sum_{\ell=n-t+1}^{j-1}\rho\left(A^{(t)}_{j,\ell}\right)\right\}. \label{eq:boundgen_rankH_inductionhyp}
\end{eqnarray}
We then need to prove that for any $t$ such that $3+(s+1) \leq t \leq n$,  rank of $H^{(t)}$ is lower bounded by
\begin{eqnarray}
\rho\left( H^{(t)} \right) & \geq & \frac{2}{(s+2)(s+3)}\left\{ (s+2)\sum_{j=n-t+1}^{n}\rho\left(A^{(t)}_{j,j}\right) - \sum_{j=n-t+2}^{n} \sum_{\ell=n-t+1}^{j-1}\rho\left(A^{(t)}_{j,\ell}\right)\right\}. \label{eq:boundgen_rankH_inductionstep}
\end{eqnarray}
Towards this first of all note that $3 + (s+1) \leq t \implies 3 + s \leq t-1$, and hence we can apply the induction hypothesis to any pair $(s, t-1), \ 3 + (s+1) \leq t \leq n$. Thus, using \eqref{eq:boundgen_rankH_inductionhyp}, the rank of the matrix $H^{(t-1)}$ can be lower bounded as\
\begin{eqnarray}
\rho\left( H^{(t-1)} \right) & \geq & \frac{2}{(s+1)(s+2)}\left\{ (s+1)\sum_{j=n-t+2}^{n}\rho\left(A^{(t-1)}_{j,j}\right) - \sum_{j=n-t+3}^{n} \sum_{\ell=n-t+2}^{j-1}\rho\left(A^{(t-1)}_{j,\ell}\right)\right\}. \label{eq:inductionhyp_tm1}
\end{eqnarray}
In the following sequence of steps, we consider \eqref{eq:inductionhyp_tm1} and we will express (or bound) summations involving $\{A^{(t-1)}_{i,j}\}$ with summations which only involve $\{A^{(t)}_{i,j}\}$. In this way, we will get a lower bound on the rank of $H^{(t-1)}$ in terms of $\{A^{(t)}_{i,j}\}$. We will then invoke the comparison $\rho(H^{(t)}) \geq \rho(H^{(t-1)})$, where $\rho(H^{(t)})$ is as before, given by \eqref{eq:nkk_proof_rankHt_5}. The bound in \eqref{eq:boundgen_rankH_inductionstep} will turn out to be a necessary condition for satisfying $\rho(H^{(t)}) \geq \rho(H^{(t-1)})$. At this point, it may be noted that the difference between the case of induction start (with $s=1$) and the case for any general $s$ lies in the expression for $\rho(H^{(t-1)})$ that we use for performing comparison $\rho(H^{(t)}) \geq \rho(H^{(t-1)})$. The expression for $\rho(H^{(t)})$ is always as given by \eqref{eq:nkk_proof_rankHt_5}.

Recall from Part c) of Lemma \ref{lem:nkk_intersections} that we have (replace $t$ with $t-1$ in Part c) of Lemma \ref{lem:nkk_intersections})
\begin{eqnarray}
\sum_{\ell=n-t+2}^{j-1} \rho\left(A^{(t-1)}_{j,\ell} \right) & \leq & \sum_{\ell=n-t+1}^{j-1}\rho\left(A^{(t)}_{j,\ell} \right) \ - \ \rho\left(A^{(t-1)}_{j,j} \right),  \ n-t+3 \leq j \leq n. \label{eq:proof_indstep_1}
\end{eqnarray}
Using \eqref{eq:proof_indstep_1} in \eqref{eq:inductionhyp_tm1}, we get 
\begin{eqnarray}
\rho\left( H^{(t-1)} \right) & \geq & \frac{2}{(s+1)(s+2)}\left\{ (s+1)\sum_{j=n-t+2}^{n}\rho\left(A^{(t-1)}_{j,j}\right) - \sum_{j=n-t+3}^{n} \left( \sum_{\ell=n-t+1}^{j-1}\rho\left(A^{(t)}_{j,\ell} \right) \ - \ \rho\left(A^{(t-1)}_{j,j} \right)\right)\right\} \label{eq:proof_indstep_2} \\ 
& = & \frac{2}{(s+1)(s+2)}\left\{ (s+1)\rho\left(A^{(t-1)}_{n-t+2,n-t+2}\right) + (s+2)\sum_{j=n-t+3}^{n}A^{(t-1)}_{j,j} - \sum_{j=n-t+3}^{n} \sum_{\ell=n-t+1}^{j-1}\rho\left(A^{(t)}_{j,\ell} \right) \right\}. \nonumber \\ \label{eq:proof_indstep_3} 
\end{eqnarray}
For proceeding further, note from \eqref{eq:nkk_proof_rankHtp1_2}-\eqref{eq:nkk_proof_rankHtp1_3az} that we have
\begin{eqnarray} \label{eq:proof_indstep_4} 
\rho\left(A^{(t-1)}_{n-t+2,n-t+2}\right) & = &  \rho\left(A^{(t)}_{n-t+2,n-t+2} \right)   \ - \ \left(\rho\left(A_{n-t+2,n-t+2}^{(t)}\right) - \rho\left(A_{n-t+2,n-t+1}^{(t)}\right)\right)^+  -   \alpha_{n-t+2}.
\end{eqnarray}
Also, one can see from \eqref{eq:nkk_proof_rankHtp1_4noplus}-\eqref{eq:nkk_proof_rankHtp1_7} that the quantity $\rho\left(A^{(t-1)}_{j,j}\right), n-t+3 \leq j \leq n$ is given by 
\begin{eqnarray} \label{eq:proof_indstep_5} 
\rho\left(A^{(t-1)}_{j,j}\right) & = & \rho\left(A^{(t)}_{j,j} \right) \ - \ \left\{ \left(\rho\left(A_{j,j}^{(t)}\right) - \sum_{\ell=n-t+1}^{j-1}\rho\left(A_{j,\ell}^{(t)}\right)\right)^+  +   \alpha_j \right\}, \ n-t+3 \leq j \leq n. 
\end{eqnarray}
Substituting for $\rho\left(A^{(t-1)}_{n-t+2,n-t+2}\right)$ and $\rho\left(A^{(t-1)}_{j,j}\right), n-t+3 \leq j \leq n$ in \eqref{eq:proof_indstep_3}, using \eqref{eq:proof_indstep_4} and \eqref{eq:proof_indstep_5} respectively, we get 
\begin{eqnarray}
\rho\left( H^{(t-1)} \right) & \geq & \frac{2}{(s+1)(s+2)}\left\{ (s+1)\left[ \rho\left(A^{(t)}_{n-t+2,n-t+2} \right)   \ - \ \left(\rho\left(A_{n-t+2,n-t+2}^{(t)}\right) - \rho\left(A_{n-t+2,n-t+1}^{(t)}\right)\right)^+  -   \alpha_{n-t+2}\right] \right. \nonumber \\
& & \hspace{-1in} \left. + \ (s+2)\sum_{j=n-t+3}^{n}\left[ \rho\left(A^{(t)}_{j,j} \right) \ - \ \left\{ \left(\rho\left(A_{j,j}^{(t)}\right) - \sum_{\ell=n-t+1}^{j-1}\rho\left(A_{j,\ell}^{(t)}\right)\right)^+  +   \alpha_j \right\} \right] \ - \ \sum_{j=n-t+3}^{n} \sum_{\ell=n-t+1}^{j-1}\rho\left(A^{(t)}_{j,\ell} \right) \right\}, \label{eq:proof_indstep_6} 
\end{eqnarray}
and thus we have our bound on $\rho\left( H^{(t-1)} \right)$ in terms of $\{A_{i, j}^{t}\}$. 

We now perform the comparison $\rho\left(H^{(t)}\right) \geq \rho\left(H^{(t-1)}\right)$. Towards, this, recall from \eqref{eq:nkk_proof_rankHt_5} that the rank of the matrix $H^{(t)}$ is given by
\begin{eqnarray}
\rho\left( H^{(t)} \right) & = & \rho\left(A^{(t)}_{n-t+1,n-t+1}\right) \ + \  \sum_{j=n-t+2}^{n} \left\{ \left(\rho\left(A_{j,j}^{(t)}\right) - \sum_{\ell=n-t+1}^{j-1}\rho\left(A_{j,\ell}^{(t)}\right)\right)^+  +   \alpha_j \right\}, \ 4 \leq t \leq n. \label{eq:nkk_proof_rankHt_5rep2}
\end{eqnarray}
One can verify that the comparison $\rho\left( H^{(t)} \right) \geq \rho\left( H^{(t-1)} \right)$ using the above two equations yields the following lower bound on $\sum_{j=n-t+2}^{n}\alpha_j$:
\begin{eqnarray}
\sum_{j=n-t+2}^{n}\alpha_j & \geq & \frac{1}{(s+2)(s+3)}\left\{ -(s+1)(s+2)\rho\left(A^{(t)}_{n-t+1,n-t+1}\right) + 2(s+1)\rho\left(A^{(t)}_{n-t+2,n-t+2}\right)   \right. \nonumber\\
& & + \ 2(s+2)\sum_{j=n-t+3}^{n}\rho\left(A^{(t)}_{j,j}\right) - (s+1)(s+4)\left( \rho\left(A^{(t)}_{n-t+2,n-t+2}\right) - \rho\left(A^{(t)}_{n-t+2,n-t+1}\right)\right)^+   \nonumber\\
& & \left. - \ (s+2)(s+3)\sum_{j=n-t+3}^{n}\left( \rho\left(A^{(t)}_{j,j}\right) - \sum_{\ell=n-t+1}^{j-1}\rho\left(A^{(t)}_{j,\ell}\right)\right)^+ - 2\sum_{j=n-t+3}^{n} \sum_{\ell=n-t+1}^{j-1}\rho\left(A^{(t)}_{j,\ell}\right)\right\} \label{eq:nkk_alphasum_indstep}.
\end{eqnarray}
It can be checked that substitution of \eqref{eq:nkk_alphasum_indstep} in \eqref{eq:nkk_proof_rankHt_5rep2} yields the desired lower bound given by
\begin{eqnarray}
\rho\left( H^{(t)} \right) & \geq & \frac{2}{(s+2)(s+3)}\left\{ (s+2)\sum_{j=n-t+1}^{n}\rho\left(A^{(t)}_{j,j}\right) - \sum_{j=n-t+2}^{n} \sum_{\ell=n-t+1}^{j-1}\rho\left(A^{(t)}_{j,\ell}\right)\right\}, 4+s \leq t \leq n.
\end{eqnarray}
This completes the proof of Theorem \ref{thm:nkk_rankviainduction}.
\end{document}